%% file: main.tex
\documentclass[9pt, notitlepage]{article}

\usepackage[letterpaper,margin=1.00in]{geometry}
\usepackage{amsmath, amssymb, amsthm, amsfonts}
\usepackage{thmtools,thm-restate}
\usepackage{bbm}

\input{macros}

\title{Deterministic Distributed Sparse and Ultra-Sparse Spanners \\ and Connectivity Certificates
}
\author{
  Marcel Bezdrighin\\
  \small{ETH Zurich} \\
  \small{mbezdrighin@student.ethz.ch}
  \and
  Michael Elkin\\
  \small{Ben-Gurion University}\\
  \small{elkinm@cs.bgu.ac.il}
  \and
  Mohsen Ghaffari \\
  \small{ETH Zurich}\\
  \small{ghaffari@inf.ethz.ch}
  \and
  Christoph Grunau \\
  \small{ETH Zurich}\\
  \small{cgrunau@inf.ethz.ch}
  \and
  Bernhard Haeupler \\
  \small{ETH Zurich and CMU}\\
  \small{bernhard.haeupler@inf.ethz.ch}
  \and
  Saeed Ilchi  \\
  \small{ETH Zurich} \\
  \small{saeed.ilchi@inf.ethz.ch}
  \and
  Václav Rozhoň \\
\small{ETH Zurich} \\
\small{rozhonv@inf.ethz.ch}
}

\date{}

\begin{document}
\maketitle
\begin{abstract}
This paper presents efficient distributed algorithms for a number of fundamental problems in the area of graph sparsification:

\begin{itemize}
\item We provide the first deterministic distributed algorithm that computes an ultra-sparse spanner in $\polylog(n)$ rounds in weighted graphs. Concretely, our algorithm outputs a spanning subgraph with only $n+o(n)$ edges in which the pairwise distances are stretched by a factor of at most $O(\log n \;\cdot\; 2^{O(\log^* n)})$. 

\item We provide a  $\polylog(n)$-round deterministic distributed algorithm that computes a spanner with stretch $(2k-1)$ and $O(nk + n^{1 + \frac{1}{k}} \log k)$ edges in unweighted graphs and with $O(n^{1 + \frac{1}{k}} k)$ edges in weighted graphs. 

\item We present the first $\polylog(n)$ round randomized distributed algorithm that computes a sparse connectivity certificate. For an $n$-node graph $G$, a certificate for connectivity $k$ is a spanning subgraph $H$ that is $k$-edge-connected if and only if $G$ is $k$-edge-connected, and this subgraph $H$ is called sparse if it has $O(nk)$ edges. Our algorithm achieves a sparsity of $(1 + o(1))nk$ edges, which is within a $2(1 + o(1))$ factor of the best possible. 
\end{itemize}

\end{abstract}



\input{introduction}

\input{preliminaries}

\input{derand}
\input{reduction}
\input{P2-300_sparse_to_low_diameter}

\input{conn-certificate}
 \section*{Acknowledgment}
M.E. was supported by the ISF grant No. (2344/19).  M.G., C.G., S. I., and V.R. were supported in part by the European Research Council (ERC) under the European Unions Horizon 2020 research and innovation program (grant agreement No.~853109) and the Swiss National Foundation (project grant 200021-184735). B.H. was supported in part by NSF grants CCF-1814603, CCF-1910588, NSF CAREER award CCF-1750808, a Sloan Research Fellowship, funding from the European Research Council (ERC) under the European Union's Horizon 2020 research and innovation program (grant agreement 949272), and the Swiss National Foundation (project grant 200021-184735).

\bibliographystyle{alpha}
\bibliography{refs}

\newpage
\appendix
{

\onecolumn
\fontsize{11pt}{14pt}\selectfont
\input{appendix}
}

\end{document}

%% file: macros.tex
\usepackage{ amssymb }
\usepackage{pifont}
\usepackage{color, colortbl}

\usepackage{amsthm,amsmath,xcolor,graphicx}
\usepackage{thm-restate}
\usepackage{bbm}
\usepackage{xcolor}
\usepackage{appendix}
\usepackage{mathbbol}
\usepackage{comment}
\usepackage[shortlabels]{enumitem}
\usepackage{comment}
\usepackage{titling}

\usepackage[textsize=tiny,backgroundcolor=orange!20]{todonotes}

\definecolor{darkgreen}{rgb}{0,0.5,0}
\usepackage{hyperref}
\hypersetup{
    unicode=false,          
    colorlinks=true,        
    linkcolor=blue,          
    citecolor=darkgreen,        
    filecolor=magenta,      
    urlcolor=cyan           
}
\usepackage{etoolbox}
\usepackage{xspace}
 \usepackage{algorithm}
\usepackage[noend]{algpseudocode}
\usepackage[labelfont=bf]{caption}
\usepackage[noabbrev,capitalize,nameinlink]{cleveref}
\crefname{equation}{}{} 
\AtBeginEnvironment{appendices}{\crefalias{section}{appendix}} 

\usepackage[color,final]{showkeys} 

\colorlet{refkey}{orange!20}
\colorlet{labelkey}{blue!30}

\numberwithin{equation}{section}
\newtheorem{theorem}{Theorem}[section]

\newtheorem{lemma}[theorem]{Lemma}
\newtheorem{claim}[theorem]{Claim}

\newtheorem{observation}[theorem]{Observation}
\crefname{claim}{Claim}{Claims}
\Crefname{remark}{Remark}{Remarks}
\Crefname{observation}{Observation}{Observations}

\newtheorem*{question*}{Question}

\theoremstyle{definition}
\newtheorem{definition}[theorem]{Definition}

\newtheorem*{definition*}{Definition}

\theoremstyle{remark}
\newtheorem*{remark}{Remark}

\newcommand{\poly}{\mathrm{poly}}

\newcommand{\fC}{\mathcal{C}}
\newcommand{\fH}{\mathcal{H}}

\newcommand{\eps}{\varepsilon}

\newcommand{\E}{\mathbb{E}}

\newcommand{\diam}{\textrm{diam}}

\newcommand{\fR}{\mathcal{R}}
\newcommand{\fB}{\mathcal{B}}

\newcommand{\fA}{\mathcal{A}}

\newcommand{\polylog}{\operatorname{\text{{\rm polylog}}}}

\newcommand{\congest}{$\mathsf{CONGEST}\,$}
\newcommand{\local}{$\mathsf{LOCAL}\,$}
\newcommand{\pram}{$\mathsf{PRAM}\,$}

\newcommand{\CONGEST}{$\mathsf{CONGEST}$\xspace}

\newcommand{\LOCAL}{$\mathsf{LOCAL}$\xspace}

\newcommand{\inv}{\mathrm{inv}}

\DeclareMathOperator{\Prob}{\mathrm{Pr}}

\newcommand{\hopdiam}{\mathrm{HopDiameter}}

\newcommand{\Tldc}{T}
\newcommand{\Dldc}{D}

\newcommand{\xiavg}{\xi_\text{AVG}}

\newcommand{\diamwldc}{\Dldc}

%% file: introduction.tex
\section{Introduction}

This paper studies \emph{graph sparsification} problems, particularly connectivity-preserving and distance-preserving problems. Concretely, given a network $G=(V, E)$, we are interested in computing a spanning subgraph $H=(V, E')$, where the subset $E'\subset E$ has as few as possible edges, while $H$ preserves some particular connectivity or distance properties of $G$, that we will soon elaborate on. Such sparsifications have various applications. The prototypical usage is that, instead of operating on the entire network $G$, we now operate on the sparser network $H$ and this would save various costs. For instance, this cost might be the monetary price paid to the network provider, the total amount of communication and thus the energy consumed, or the overall computational complexity. All of these can be proportional to the number of edges in use, and thus operating on a sparser network saves costs. 

The two families of properties that we focus on in this paper are  \emph{distances} and \emph{connectivity}. The former concept is usually called \emph{spanner}, while the latter is often studied under titles such as \emph{sparse $k$-connectivity certificates} and \emph{approximation algorithms for $k$-edge-connected subgraphs}. Below, we discuss two categories separately, in each case first reviewing the definitions and the state of the art and then stating our contributions. Before that, let us briefly recall the standard message-passing model of distributed computation, which we use throughout.

\paragraph{Distributed Model.} We work with the standard synchronous message-passing model of distributed computation, often called the \CONGEST model~\cite{peleg2000distributed}. Here, the network is abstracted as an $n$-node graph $G=(V, E)$ where each node represents one computer/processor in the network. Each processor/node has a unique $O(\log n)$-bit identifier. Communications take place in synchronous rounds, where per round each node can send one $O(\log n)$-bit message to each of its neighbors. We comment that if the model is relaxed to allow unbounded message sizes, then it is referred to as the \LOCAL model. Throughout the paper, we will work with the \CONGEST model but we will discuss in the related work some results in the relaxed \LOCAL model.
Generally, the input and output are represented in a distributed fashion. At the start of the algorithm, the processors/nodes do not know the topology of the network $G$; each of them knows just its own neighbors and is able to communicate with each of those neighbors once per round. At the end of the algorithm, each node should know its own part of the output, e.g., when discussing sparsification problems, each node should know which of its edges is the computed subgraph $H$ of $G$.

\subsection{Spanners: Background}
\label{subsec:spanners-background}
Graph spanners were introduced by Peleg and Sch{\"{a}}ffer \cite{peleg1989graph}. Since then, spanners and efficient distributed, parallel, and sequential constructions for them have been studied extensively~\cite{althofer1993sparse,AwerbuchBCP93,cohen93,AingworthCIM99,DorHZ00,elkin20041eps,elkin2005computing,thorup2006spanners,baswana2007simple,Pettie09,pettie2010distributed,BaswanaKMP10,chechik2013new,AbboudB15,elkin2018efficient,grossman2017improved,ghaffari2018derandomizing,censor2018distributed,ParterY18,elkin2019near}. 

\begin{definition}[$\alpha$-spanner]
For a graph $G = (V,E)$, a subgraph $H=(V,E')$ is a \emph{spanner with stretch} $\alpha>1$ ---or simply an $\alpha$-\textit{spanner}--- if and only if $d_{H}(u,v)\leq \alpha\cdot d_G(u,v)$ for all $u,v \in V$. We assume that $G$ is undirected, and we are interested in minimizing $|E'|$ in terms of integers $n$ and $\alpha$.
\end{definition}

Of particular interest in this paper are \emph{sparse spanners}, which have $O(n)$ edges, and \emph{ultra-sparse spanners}, which have $n + o(n)$ edges. More concretely, an ultra-sparse spanner of an $n$-vertex graph has at most $n + n/t$ edges for some $t>1$. That is, it is made of $n-1$ edges---as necessary for a spanning tree---and only $n/t +1 $ extra edges. As $t$ gets larger, the structure gets closer and closer to a spanning tree, which is the minimal subgraph to keep the graph connected.

\paragraph{Motivation and Applications.} Spanners have found numerous applications, including in packet routing, constructing synchronizers,  and algorithmics of various graph problems. See e.g.,~\cite{AwerbuchBCP93,AwerbuchP92,cohen93,thorup2005approximate,ThorupZ01routing,GavoillePPR04,LenzenL18} and the recent survey of Ahmed et al.~\cite{ahmed2020graph}. 

Again, of particular interest in this paper are sparse and ultra-sparse spanners. A direct application of sparse and ultra-sparse spanners is that we can consider them as a sparse skeleton, i.e., a subgraph that retains connectivity, with an asymptotically optimal number of edges. Furthermore, ultra-sparse spanners have found important applications in the context of solving symmetrically diagonally-dominant (SDD) linear systems, primarily as ultra-sparse spanners can be seen as one spanning tree and only few extra edges. The applications include 
estimating effective resistances, and computing spectral sparsifiers or ultra-sparsifiers for cuts and flows~\cite{SS11,BGKMPT14}. 
Since cuts, flows, and distances are easy in trees and since ultra-sparse spanners are trees together with a very small number of extra edges, these ultra-sparse structures have emerged as a powerful tool for (recursively) reducing the complexity of numerous fundamental optimization and graph problems. These include maximum flow~\cite{peng2016approximate}, min-cost flow, lossy flow problems, several variants of min-cut problems~\cite{DS08}, as well as approximate shortest paths and transshipment problems~\cite{Li20}.  

\paragraph{State of the art, small stretch spanners. }
We first review the state of the art for spanners in the general case of $k$ and then focus on the particularly interesting regime of sparse and ultra-sparse spanners. 

Alth{\"o}fer et al.~\cite{althofer1993sparse} provided a simple greedy algorithm for constructing $(2k-1)$-spanner with $n^{1 + \frac{1}{k}}$ edges, in unweighted and weighted graphs. This size bound is tight conditioned on the Erd\H{o}s girth conjecture: There is a family of graphs with girth $2k+2$ and $\Omega(n^{1+1/k})$ edges~\cite{erdos1963extremal}. 
Note that by setting $k = \omega({\log n}) $, one obtains ultra-sparse spanners with stretch $\omega(\log n)$.

In the distributed setting, we discuss prior results in two parts of randomized and deterministic algorithms. Baswana and Sen~\cite{baswana2007simple,pettie2010distributed} presented an $O(k)$ rounds randomized algorithm for $(2k-1)$-stretch unweighted spanners with expected size $O(nk + n^{1 + \frac{1}{k}} \log k)$.\footnote{In the original paper~\cite{baswana2007simple}, it is claimed that the spanner has $O(nk + n^{1 + \frac{1}{k}})$ edges. Pettie~\cite{pettie2010distributed} noticed a gap in their argument and provides a bound but with an extra $\log k$ factor. The original claim  remains unproven.} For weighted graphs, their output has $O(n^{1 + \frac{1}{k}} k)$ edges. Inspired by the work of Miller et al.~\cite{miller2015improved}, Elkin and Neiman~\cite{elkin2018efficient} gave a randomized $(2k-1)$-spanner for unweighted graphs that runs in $O(k)$ rounds and matches the centralized bound of $O(n^{1 + \frac{1}{k}})$ with a constant probability; concretely, for any $\eps>0$, it has size $O(n^{1 + \frac{1}{k}}/\eps)$ with probability at least $1-\eps$. 

For deterministic algorithms, the work of Barenboim et al.~\cite{barenboim2018fast} gave an unweighted spanner with size $O(n^{1 + \frac{1}{k}})$ but with the weaker stretch $O(\log^{k-1} n)$ in $(\log^{k-1} n)$ rounds. Derbel et al.~\cite{derbel2010sublinear} compute $(2k-1)$-stretch unweighted spanners with $O(n^{1+1/k})$ edges but with a rather high round complexity of $O(n^{1 - \frac{1}{k}})$. Grossman and Parter~\cite{grossman2017improved} attain the same stretch with size $O(k n^{1+1/k})$ and improved round complexity of $O(2^k n^{\frac{1}{2} - \frac{1}{k}})$. The first $\polylog(n)$ rounds algorithm for $(2k-1)$-spanners in unweighted graphs was devised by Ghaffari and Kuhn~\cite{ghaffari2018derandomizing}. Their output has size $O(nk + k n^{1 + \frac{1}{k}} \log n)$. If we allow unbounded message sizes, there is a deterministic distributed algorithm by Derbel et al.~\cite{derbel2008locality} with $O(kn^{1 + \frac{1}{k}})$ edges and in $O(k)$ rounds in the \LOCAL model.

\paragraph{State of the art, sparse and ultra-sparse spanners.} The problem of devising efficient distributed and parallel algorithms for computing ultra-sparse spanners in unweighted graphs was extensively studied in~\cite{derbel2010sublinear,pettie2010distributed,ram2011distributed,barenboim2018fast,elkin2018efficient}. Pettie \cite{pettie2010distributed} presented a distributed randomized algorithm for computing a $O\left(2^{\log^* n} \log n\right)$-spanner with $O(n)$ edges in $O\left(2^{\log^* n} \log n\right)$ rounds. As discussed before, the randomized algorithm of Elkin and Neimann \cite{elkin2018efficient} provides a $(2k-1)$-spanner with $O(n^{1 + \frac{1}{k}})$ edges. With $k = \Theta(\log n)$, this automatically gives an $O(n)$-size sparse spanner. Indeed, their algorithm can compute, with a constant probability, a spanner with $O(t\log n)$-stretch and $n+n/t$ edges, in $O(t\log n)$ rounds. For a discussion on the known centralized and parallel approaches to ultra-sparse spanners, see \Cref{app:otherRelated}. The state of the art parallel algorithm is that of  Li~\cite{Li20}, and our results improve on it as we mention later.

Unfortunately, the above distributed algorithms do not provide ultra-sparse spanners in weighted graphs. 
We note that for many of the modern applications of ultra-sparse spanners in algorithmic problems mentioned above we need ultra-sparse spanners for weighted graphs. 

We comment that there is a standard and simple reduction from weighted graphs to unweighted graphs, but this reduction does not provide ultra-sparse or even sparse spanners. The reduction works as follows: to compute a spanner for a weighted graph $G$, we round the weight of each edge to a multiple of $(1+\eps)$, and then compute an unweighted $\alpha$-spanner for the edges of each of the $O(\frac{\log (U+1)}{\eps})$ weight classes separately, where $U$ is the aspect ratio of the edge weights. The union of these spanners forms a spanner with stretch $(1+\eps) \alpha$ for $G$. Besides the $(1+\eps)$-factor loss in stretch, this reduction loses an $O(\frac{\log (U+1)}{\eps})$-factor in the number of edges. For ultra-spanners, this reduction completely destroys the ultra-sparsity as the union of even just two ultra-spanners is no longer ultra-sparse. Also, in the standard cases of weighted graphs, we usually assume $\log U = O(\log n)$, and thus this reduction does not give even a sparse spanner with $O(n)$ edges.

\paragraph{Lower bounds for ultra-sparse spanners} It is known that for every $n$ and $t$, there exists some graph for which any spanner with at most $n + n/t$ edges has stretch at least $\Omega(t \log n)$. Furthermore, computing any spanner with $n + n/t$ edges requires $\Omega(t \cdot \log n)$ rounds of distributed computation~\cite{elkin2007near,derbel2009local}. Indeed, these lower bounds hold even for unweighted graphs, randomized computations, and the \local model. The \local model is much more permissive than the \congest model considered in this paper in that it allows nodes to exchange messages of unbounded size and perform arbitrarily complex local computations.


\subsection{Spanners: Our Contribution} 

As a first-order summary of prior work, to the best of our knowledge, there are no known deterministic distributed algorithms for ultra-sparse or even sparse spanners with $O(n)$ edges, even if we restrict ourselves to unweighted graphs. Furthermore, even for spanners of higher density and small stretch $2k-1$ for $k \leq \log n$, the number of edges achieved by the state of the art deterministic algorithm~\cite{ghaffari2018derandomizing} is higher than the corresponding randomized algorithms~\cite{baswana2007simple}. 

Our contributions resolve this situation: 
\begin{enumerate}
    \item We present deterministic algorithms that compute ultra-sparse spanners with $n+o(n)$ edges and close to $O(\log n)$ stretch, even in weighted graphs. We do this by showing a reduction from the ultra-sparse case to sparse spanners, and by derandomizing a randomized sparse spanner construction of Pettie~\cite{pettie2010distributed}, and extending it to weighted graphs (see \cref{table:ultrasparse}).  
    \item For general stretch parameter $k$, we provide a derandomization of Baswana-Sen~\cite{baswana2007simple} that matches its stretch-size tradeoff and improves on~\cite{ghaffari2018derandomizing} (see \cref{table:baswana}). 
    \item These derandomization-based algorithms use local computations that exceed $O(m\poly\log n)$ and thus are less suitable for adaptation to the PRAM model of parallel computation. To avoid that, we show in addition a work-efficient reduction from spanners to so-called weak-diameter clusterings. Using this connection and our reduction from ultra-sparse spanners to spanners of higher density, we achieve a work-efficient distributed/parallel algorithm for ultra-sparse weighted spanners with $\poly(\log n)$ stretch. 
\end{enumerate}
In the next three subsections, we elaborate on these contributions.

\definecolor{Gray}{gray}{0.9}

\begin{table*}[h!]
\centering
\begin{tabular}{||c | c | c | c | c | c||} 
 \hline
 Paper & Stretch & Size & Weighted? & Deterministic? & \# rounds \\ [0.5ex] 
 \hline\hline
 \cite{pettie2010distributed} & $O(\log n \;\cdot\;  2^{\log^* n})$ & $O(n)$ &$\times$& $\times$& $O(\log n) \cdot  2^{\log^* n}$ \\ 
  \hline
\cite{elkin2018efficient} & $O(\log n)$ & $n + o(n)$ expected & $\times$& $\times$ & $O(\log n)$ \\ 
\rowcolor{Gray}
  \hline
This paper & $O(\log n\;\cdot\; 2^{O(\log^* n)})$ & $n + o(n)$ & \checkmark & \checkmark & $\poly(\log n)$ \\ 
  \hline
\end{tabular}
\caption{This table summarizes the relevant results for construction of spanners that are very sparse. Ideally, we aim for ultra-sparse spanners with size $n + o(n)$. This property is crucial in some applications, as discussed in the introduction. We highlight that our result is the first ultra-sparse spanner construction that is deterministic. Our result is also the first that handles weighted graphs. }
\label{table:ultrasparse}
\end{table*}

\begin{table*}[h!]
\centering
\begin{tabular}{||c | c | c | c | c | c||} 
 \hline
 Paper & Stretch & Size & Weighted? & Deterministic? & \# rounds \\ [0.5ex] 
 \hline\hline
 \cite{baswana2007simple} & $2k-1$ & $O(n^{1 + 1/k}\log k + nk)$ & $\times$ & $\times$& $O(k)$ \\ 
   \hline
\cite{baswana2007simple} & $2k-1$ & $O(n^{1 + 1/k}k)$ & \checkmark & $\times$& $O(k)$ \\ 
    \hline
\cite{ghaffari2018derandomizing} & $2k-1$ & $O(n^{1 + 1/k}k\log n)$ &$\times$ & \checkmark & $\poly(\log(n))$ \\ 
  \hline
\rowcolor{Gray}
This paper & $2k-1$ & $O(n^{1 + 1/k}\log k + nk)$ & $\times$ & \checkmark & $\poly(\log(n))$ \\ 
  \hline
\rowcolor{Gray}
This paper & $2k-1$ & $O(n^{1 + 1/k}k)$ & \checkmark & \checkmark & $\poly(\log(n))$ \\ 
  \hline
\end{tabular}
\caption{This table summarizes the relevant results for construction of spanners with small stretch (think of $k$ as a constant). We highlight that our results match the bounds of \cite{baswana2007simple} while our results are deterministic. }
\label{table:baswana}
\end{table*}

\subsubsection{Reduction from Ultra-Sparse Spanners to (Sparse) Spanners}
Our first contribution is a general reduction from the problem of computing ultra-sparse spanners in weighted graphs to the problem of computing their spanners. 
This allows us to efficiently move the extra factor in the number of edges of the final spanner to its stretch:

Before we state the theorem, we briefly introduce the notion of a \emph{cluster graph}: Given a graph $G$, a cluster graph $H$ is a graph that we get by contracting some of its disjoint subgraphs that we call clusters. 
If all of those clusters are connected and have radius at most $r$, we say that $H$ is an $r$-cluster-graph of $G$. We say that a distributed algorithm works on an $N$-node $r$-cluster-graph in $T(N,r)$ rounds if the underlying communication network is $G$, but the output of the algorithm should be for a given $r$-cluster-graph $H$ of $G$ with $|V(H)| = N$. For a more precise definition of a cluster graph and an algorithm that operates on a cluster graph see \Cref{sec:prelim}. 

\begin{restatable}{theorem}{ThmSparseToUltraSparse}
\label{thm:p1-sparse-ultrasparse}
Suppose that there exists a deterministic distributed algorithm $A$ which computes an $\alpha$-spanner with $N \cdot s(N)$ edges for any $N$-node weighted $r$-cluster-graph in $T(N,r)$ rounds. 
Then, for any $t \geq 1$, there is a deterministic distributed algorithm $A'$ that computes an $O(t\cdot s(n) \cdot \alpha)$-spanner with $n + n/t$ edges for any $n$-node weighted graph, in 
\[
O\left(t\cdot s(n) \cdot \log^* n + T\left(\frac{n}{t\cdot s(n)}, O\left(t\cdot s\left(n\right)\right)\right)\right)
\]
rounds.
\end{restatable}
We present the proof of \Cref{thm:p1-sparse-ultrasparse} in \Cref{sec:reduction}.
We note that the stretch-size tradeoff in our reduction is asymptotically optimal. Indeed, plugging an $O(\log n)$-spanner construction with $O(n)$ edges~\cite{derbel2009local,elkin2018efficient} in our reduction would result in an ultra-spanner with $n+n/t$ edges and stretch $O(t \log n)$, which is asymptotically optimal, even in unweighted graphs. 
 
\paragraph{Example Implication.}
As one concrete corollary of \cref{thm:p1-sparse-ultrasparse}, let us discuss how this reduction gives an improvement on the results of Li~\cite{Li20}. By applying \Cref{thm:p1-sparse-ultrasparse} to the Pettie's~\cite{pettie2010distributed} randomized sparse spanner construction, we obtain a randomized algorithm for weighted ultra-sparse spanner with $n+n/t$ edges and $O(t \log n \cdot 4^{\log^* n})$ stretch. More concretely, we have:

\begin{restatable}{theorem}{ThmRandomizedUltraSparse}
\label{thm:randomized-ultrasparse}
There is a randomized distributed algorithm that computes an $O\left(t\log n \cdot 4^{\log^* n} \right)$-spanner of any $n$-node weighted graph and any $t \ge 1$ with expected $n + n/t$ edges. The algorithm works with high probability and it runs in $t\cdot \polylog(n)$ rounds in the \CONGEST model. A parallel variant of this algorithm works with $\polylog(n)$ depth and $m \cdot \polylog(n)$ work in the \pram model.
\end{restatable}

This stretch is optimal up to the $4^{\log^* n}$ term. We comment that Pettie presents the algorithm only for unweighted graphs, but we show in \cref{thm:linearsize} that by simple modifications, we can extend the algorithm to weighted graphs, with only a minimal loss of increasing the stretch factor from $O(\log n \cdot 2^{\log^* n})$ to $O(\log n \cdot 4^{\log^* n})$. 
Also, as the algorithm readily runs in the \pram model with $\polylog(n)$ depth and $m \cdot \polylog(n)$ work. 
This improves on the \pram ultra-sparse spanner construction of Li~\cite{Li20}, which had stretch $O(t^2 \log^3n \log^2 \log n)$ for $n+n/t$ edges.

In the next two subsubsections we discuss how invoking the reduction of \cref{thm:p1-sparse-ultrasparse} atop our (deterministic) spanner constructions leads to our final results on ultra-sparse spanners.

\subsubsection{Spanners via Derandomization}
Our second contribution is derandomization-based deterministic spanner constructions (see \Cref{sec:derand}). We present these in two groups: (A) focused on stretch and especially for sub-logarithmic values of stretch, and (B) focused on low sparsity for roughly logarithmic values of stretch. 
We note that the algorithms below achieve a good stretch for the given number of edges, but they involve large computations in each node.

\smallskip
\paragraph{(A) Low stretch spanners, by derandomizing Baswana-Sen~\cite{baswana2007simple}:}  In the first category, our focus is on small values of stretch. We show  $\polylog(n)$-round deterministic algorithms for $(2k-1)$-spanners, with $O(nk+n^{1+1/k}\log k)$ edges in unweighted graphs and $O(n^{1+1/k} k)$ edges in weighted graphs:

\begin{restatable}{theorem}{ThmKStretch}
\label{thm:kstretch}
There are $\polylog(n)$-rounds deterministic distributed algorithms that compute a $\left(2k-1\right)$-spanner with $O(nk + n^{1 + 1/k} \log k)$ and $O(nk + n^{1 + 1/k} \log k)$ edges for unweighted and weighted graphs, respectively.
\end{restatable}

 Our approach is inspired by the work of Ghaffari and Kuhn~\cite{ghaffari2018derandomizing} on derandomizing the Baswana-Sen randomized algorithm. Ghaffari and Kuhn~\cite{ghaffari2018derandomizing} get the same stretch but with a worse sparsity bound of $O(n^{1+1/k}k \log n)$ edges, and only for unweighted graphs. In contrast, our bounds match the best known analysis of the randomized Baswana-Sen algorithm, in both unweighted and weighted graphs. 
 
 \smallskip
 \paragraph{(B) Low sparsity spanners, by Derandomizing Pettie~\cite{pettie2010distributed}:} In the second category, our focus is on the sparsity of the spanner, and we show $\polylog(n)$-round deterministic algorithms that compute sparse spanners with $O(n)$ edges, in both unweighted and weighted cases, which achieve a stretch of almost $O(\log n)$:
 
\begin{restatable}{theorem}{ThmLinearSize}
\label{thm:linearsize}
There are $\polylog(n)$ rounds deterministic distributed algorithms that compute a spanner with $O(n)$ edges. For unweighted graphs, the stretch of the spanner is $O\left(\log n \cdot 2^{\log^* n}\right)$ and for weighted graphs, it is $O\left(\log n \cdot 4^{\log^* n}\right)$.
\end{restatable}
 
 These deterministic algorithms are also obtained via derandomization, but when applied to Pettie's randomized algorithm~\cite{pettie2010distributed}. We also note that the weighted sparse spanner stated in \Cref{thm:linearsize} was not stated in prior work even for randomized algorithms; we obtain this result via a minor modification of Pettie's unweighted approach.

\paragraph{Implications.} Via \Cref{thm:p1-sparse-ultrasparse} with \Cref{thm:linearsize}, we get the first $\polylog(n)$-round deterministic distributed algorithm for ultra-sparse spanners:

\begin{restatable}{theorem}{ThmUltraSparse}
\label{thm:ultrasparse}
There are $t\cdot \polylog(n)$ distributed deterministic algorithms that compute a spanner with $n + n/t$ edges with stretch $O\left(t\log n \cdot 2^{\log^* n}\right)$ for unweighted graphs and with stretch $O\left(t\log n \cdot 4^{\log^* n}\right)$ for weighted graphs.
\end{restatable}

 \subsubsection{Spanners via Low-Diameter Clusterings}

As our third contribution, we show a connection with the so-called weak-diameter clusterings. 
As a result of this connection, we will achieve deterministic spanner constructions that are work efficient---i.e., doing only $m\poly\log n$ computations---but have weaker bounds on the stretch; the latter is merely due the sub-optimal bounds in the state of the art radius for deterministic low-diameter clusterings.

Our most general result in this direction is the following near-optimal end-to-end reduction of spanner construction to clustering construction. 
Before we state the somewhat technical theorem, we need to briefly introduce the notion of weak-diameter and separated clusterings. 
A clustering is a collection of disjoint vertex sets (clusters).
Next, think actually of every cluster as a pair of the vertex set $C \subseteq V(G)$ and a tree $T_C$ with $V(T_C) \supseteq C$. That is, each cluster also carries a tree with it collecting its nodes. 
A clustering has weak-diameter $D$ if every cluster $C$ of it has the property that the diameter of $T_C$ is at most $D$. 
Namely, we require any two nodes of $C$ to be close in $G$, but the cluster $C$ itself might be even disconnected. 
A clustering is $k$-separated if the distance of any two different clusters of it is at least $k$. 
Finally, the average overlap $\xiavg$ of a clustering $\{C_1, \dots, C_t\}$ is defined as $\xiavg = \left( \sum_{i = 1}^t |V(T_{C_i})| \right) / n$. That is, it is the average overlap of the trees of the clustering. 

\begin{restatable}{theorem}{thmDeterministic}
\label{thm:final-result-weak-clustering-to-utrasparse-unweighted}
Suppose there exists an algorithm $A$ which for any unweighted $n$-vertex graph finds a $3$-separated clustering with weak-diameter $\diamwldc(n)$ with average overlap at most $\xiavg(n)$. 
Then, there is an algorithm $A'$ that builds a $\beta$-spanner $H$ on an unweighted graph such that  (1) $|E(H)| = O(\xiavg(n) n)$, (2) $\beta = O(D(n))$.

Furthermore, if $A$ is deterministic then so is $A'$ and
\begin{itemize}
        \item If $A$ requires at most $\Tldc(n)$ \congest rounds then $A'$ requires $O\left( \log(n)\Tldc(n) \right)$ \congest rounds.
        
        \item If $A$ requires at most $\Tldc(n,r)$ \congest rounds on a $r$-cluster-graph, then $A'$ requires $O\left( \log(n)(\Tldc(n,r) + r) \right)$ \congest rounds on a $r$-cluster-graph. 

        \item If $A$ is a \pram algorithm with $T(n)$ depth and $W(m,n)$ work then $A'$ is a \pram algorithm with depth $ O(\log(n)T(n))$ and work $O(\poly(\log(n))(W(m,n) + m + n \xiavg(n))$.
        
\end{itemize}
\end{restatable}

 \smallskip 
 
\Cref{thm:final-result-weak-clustering-to-utrasparse-unweighted} is optimal, up to constants, in terms of its density-stretch tradeoff. This is because there are $3$-separated clusterings with diameter $\log n$ and constant overlap. These would give the desired stretch factor of $O(\log n)$.
Also, \Cref{thm:final-result-weak-clustering-to-utrasparse-unweighted} can be extended to produce spanners in weighted graphs with an additional $O(\log (U + 1))$ overhead in the number of edges in the spanner, by applying the standard reduction mentioned earlier in \Cref{subsec:spanners-background}.

\paragraph{Implications.} Plugging the state-of-the-art deterministic distributed clusterings~\cite{rozhovn2020polylogarithmic} into \Cref{thm:final-result-weak-clustering-to-utrasparse-unweighted}, and then applying \Cref{thm:p1-sparse-ultrasparse} on top, gives the following result:

\begin{restatable}{theorem}{thmWeightedUltra}\label{thm:result-weighted-ultra-randomized-congest}
There is a deterministic work-efficient \congest algorithm that, given any $n$-vertex weighted graph $G$ and $t \ge 1$, computes in $O(t \log^{10} n)$ rounds an ultra-sparse spanner with $n + n/t$ edges and stretch $O(t \log^4 n \log (U+1))$, where $U \geq 1$ is the aspect ratio of the weights. There is also a deterministic \pram algorithm that computes such a spanner in $\polylog(n)$-time and with $m\cdot \polylog(n)$ work.
\end{restatable}




\subsection{Connectivity Certificates: Background and Our Contribution}
\paragraph{Definition and Motivation.} For a graph $G=(V, E)$, a \emph{$k$-connectivity certificate} is a spanning subgraph $H=(V, E')$ such that if $G$ is $k$-edge-connected, so is $H$. 
This high connectivity property can be relevant for resilience against link failures or for higher communication capacity. Concretely, if $G$ can withstand the failure/removal of any $k$ of its edges and it would still remain connected, the same should be true also for $H$. Similarly, if the minimum cut in $G$---which is in some sense the communication bottleneck in the network as it is the smallest number of edges between one set of nodes and the rest---has size $k$, the same should be true for $H$. 

Any $k$-edge-connected graph must have at least $nk/2$ edges, as each node must have degree at least $k$. Hence, the sparsest that the connectivity certificate $H$ can be is to have $nk/2$ edges. Considering this, any connectivity certificate that has $O(nk)$ edges is called a \emph{sparse connectivity certificate}. We note that an approximation version of this problem has also been widely studied under the notion of $k$-edge-connected spanning subgraph ($k$-ECSS), where for each given $k$-edge-connected graph $G$, the objective is to compute the sparsest possible $k$-edge-connected subgraph $H$, and the performance is measured in terms of the ratio of the number of edges in the computed subgraph $H$ to the smallest possible.

\paragraph{State of the Art.} Thurimella~\cite{thurimella1997sub} gave the first distributed algorithm for computing a sparse connectivity certificate, with $k(n-1)$ edges, and the round complexity of $O(k(D+\sqrt{n}))$ in the \CONGEST model. Here, $D$ denotes the network diameter. Primarily coming from the side of the $k$-ECSS problem, Censor-Hillel and Dory~\cite{censor2020fast} investigated the special case of $k=2$ and they gave an $O(D)$ round randomized distributed algorithm that computes a $2$-connectivity certificate with $O(n)$ edges with high probability, i.e., an $O(1)$ approximation for $2$-ECSS. Then, Dory~\cite{dory2018distributed} provided an $O(D\log^3 n)$ round randomized distributed algorithm that computes a $3$-connectivity certificate with $O(n\log n)$ edges with high probability, thus an $O(\log n)$ approximation for $3$-ECSS. Daga et al.~\cite{daga2019distributed} improved the algorithm of Thurimella~\cite{thurimella1997sub} to achieve a round complexity of $\tilde{O}(D+\sqrt{nk}))$.

Finally, Parter~\cite{parter2019small} improved the complexity significantly by providing an $O(k \poly(\log n))$-round randomized algorithm that computes a $k$-connectivity certificate with $O(k n)$ edges, with high probability. Thus, this also gives an $O(1)$ approximation for the $k$-ECSS problem in $O(k \poly(\log n))$ rounds. Notice that this complexity can still grow larger even up to $n\poly(\log n)$ as $k$ grows.
Parter~\cite{parter2019small} also gave a faster $\poly(\log n)$-round algorithm but for a slightly weaker notion of \emph{approximate}-certificate. Concretely, given the parameter $k$, the algorithm runs in $\poly(\log n)/\eps^2$ rounds and computes a spanning subgraph $H$ with $O(k n)$ edges such that, if $G$ is $k$-edge-connected, then $H$ is $k(1-\eps)$-edge-connected, with high probability. It remained open whether the same $\poly(\log n)$ round complexity can also suffice for sparse connectivity certificates, without reducing the connectivity value to the approximate version. 

\paragraph{Our contribution.} We present a simple algorithm that resolves the above question. Concretely, we show a $\poly(\log n)$-round randomized distributed algorithm that computes a $k$-connectivity certificate with $(1+o(1))kn$ edges, with high probability. 

\begin{theorem}
\label{thm:certificate-largek}
For any $\eps < 1/2$, there is a randomized distributed algorithm that computes a $k$-connectivity certificate in $\frac{\polylog(n)}{\eps^3}$ rounds with $k n(1 + \eps)$ edges.
\end{theorem}
This sparsity is within a $2+o(1)$ factor the best possible as any $k$-connectivity certificate needs at least $k n/2$ edges. Hence, our algorithm gives a $2+o(1)$ approximation for the (unweighted) $k$-ECSS problem in $\poly(\log n)$ rounds. Due to the space limitations, the entire proof of \Cref{thm:certificate-largek} is deferred to \Cref{sec:conn-certificate}. 

We note that, similar to Daga et al.~\cite{daga2019distributed} and Parter~\cite{parter2019small}, we also use Karger's edge-sampling~\cite{karger1999random} to split the graph into many edge disjoint parts and paralellize the work. Daga et al.~\cite{daga2019distributed} perform a tree packing in each of these parts. Parter packs sparse spanners, in each part, but loses a $1+\eps$ in the connectivity. We also used packings of ultra-sparse spanners, and we show by a simple case analysis cut sizes that we reach the exact $k$-connectivity, without the $1+\eps$ factor loss. We note that this $1+\eps$ factor loss can be important in applications, e.g., in the framework of Daga et al.~\cite{daga2019distributed} for computing min-cut, this loss would downgrade exact min-cut algorithms to $1+\eps$ approximation, which is a much easier problem.

%% file: preliminaries.tex
\section{Basic Notations}
\label{sec:prelim}
The input is a graph $G$, with $n$ nodes and $m$ edges. If $G$ is weighted, we assume all edge weights are non-negative and are bounded by $\poly(n)$. We denote via $d_{G}(u, v)$ the weight (i.e., length) of the shortest path between two nodes $u, v \in V(G)$. We may drop the subscript when the graph is clear from context. The distance function extends to sets by $d(A, B) = \min_{a \in A, b \in B} d(a,b)$ and we write $d(u,S)$ instead of $d(\{u\}, S)$. We denote by $\diam(G)$ the diameter of a graph $G$ which is the maximum pairwise distances between nodes of $G$. We may use $\diam(C)$ to denote the diameter of the induced subgraph $G[C]$.

For a weighted graph $G$, a \emph{cluster} $C \subseteq V(G)$ is simply a subset of its nodes. A \emph{clustering} $\fC$ is a set of disjoint clusters. We say that a clustering $\fC$ is a \emph{partition} if $V(G) = \bigcup_{C \in \fC} C$. The clustering (partition) $\fC$ is an $r$-clustering ($r$-parition) if there is a rooted tree with radius at most $r$ in the subgraph induced by each of its cluster. A rooted tree has radius $r$ if the maximum hop-distance (i.e. the number of edges in the shortest path) between a leaf and the root is $r$. We say an edge is a \textit{boundary-edge} of a cluster $C$ if exactly one of its endpoints is in $C$. An edge is an \textit{inside-edge} of $C$ if both of its endpoints are in $C$. A graph $H$ is called an $r$-cluster-graph of $G$ if it is obtained from $G$ by contracting each cluster of an $r$-clustering to a single node. If $v$ is a node of $H$, then $\inv^{H \rightarrow G}(v)$ represents the set of nodes that are contracted to $v$. Let us emphasize that in the definition of $r$-clustering and $r$-cluster-graphs, the parameter $r$ does not depend on edge weights and only depends on hop-distances.

%% file: derand.tex
\section{Derandomization}
\label{sec:derand}

\paragraph{Baswana-Sen Algorithm~\cite{baswana2007simple}.} The spanner is constructed in $k$ iterations. At first, all nodes and edges of $G$ are \textit{alive}. The input of iteration $i$ is a $(i-1)$-partition of the nodes that remain alive after the first $i-1$ iterations. During one iteration, some nodes die. When a node dies, all of its incident edges die as well. Moreover, some edges incident to an alive node may die during one iteration. By adding a relatively small set of edges to the spanner in iteration $i$, we ensure that all edges that die in this iteration have  a stretch at most $2i-1$. We also ensure that all nodes die after the last iteration. So all edges die and as a result, we have a spanner with stretch $2k-1$ at the end. The output of iteration $i$ is the input of iteration $i+1$. The input for the first iteration is the trivial partition of all nodes (one cluster for each node). The details of iteration $i$ are given in the following:

\begin{enumerate}[(1)]
    \item We sample each cluster with probability $p = n^{-1/k}$ for $i \leq k-1$. In the last iteration (when $i=k$), we sample each cluster with probability zero, i.e. no cluster is sampled.

    \item Each node adds some (possibly zero) edges to the spanner. For a node $v$ with $d$ adjacent clusters, let $e_i$ be the edge with the smallest weight $w_i$ among all edges between $v$ and its $i$-th adjacent cluster $C_i$ and assume $w_1 \leq w_2 \leq \dots \leq w_d$. A cluster $C$ is adjacent to $v$ if $v$ has a neighbor in $C$. If $v$ is in a sampled cluster, it does nothing. If $v$ is an unsampled cluster, let $i$ be the smallest integer such that $C_i$ is sampled. If there is such an $i$, node $v$ joins $C_i$. The edge $e_i$ along with all $e_j$ for which $w_j$ is strictly less than $w_i$ are added to the spanner. If there is no such an $i$, node $v$ dies and all $d$ edges $e_1, \dots, e_d$ are added to the spanner. In all of the above cases for $v$, whenever we add an edge $e_i$ to the spanner, all edges between $v$ and $C_i$ die. 
    
    \item The output of iteration $i$ is an $i$-partition on the set of nodes that are alive after iteration $i$. The partition has one cluster $C'$ for each sampled cluster $C$ along with all nodes that are joined to it. When a node $v$ joins $C$, its parent in $C$ is the node to which it has an edge with smallest weight (note that this edge is in the spanner from step (2)). Observe that the radius of $C'$ is at most the radius of $C$ plus one.
\end{enumerate}

In the following lemma we collect deterministic properties of the above construction proven in \cite{baswana2007simple}. 

\begin{lemma}[\cite{baswana2007simple}]
\label{lem:bs-stretch}
For any cluster $C$ in the output of iteration $i$, we have: Radius of $C$ is at most $i$. For any alive boundary-edge $\{u \not \in C,v \in C\}$ of $C$ with weight $w$, all edges in the unique path from $v$ to the root of $C$ have weight at most $w$. For any alive inside-edge $\{u \in C, v\in C\}$ of $C$ with weight $w$, all edges in the unique path between $u$ and $v$ in $C$ have weight at most $w$. 
It holds that all dead edges in iteration $i$ have stretch $2i-1$. So the final spanner has stretch $2k-1$ as all nodes die in the last iteration. All these properties are deterministic in the sense that they hold regardless of the way we sample clusters.
\end{lemma}

The lemma above provides a stretch guarantee we need. It remains to show the expected size of final spanner. For that, let us first define a hitting-event. We write all the stretch analysis in terms of these kind of events as we need this for derandomization. We discuss the reason in a moment. 

\begin{definition}
A binary random variable $E$ is a \textit{hitting-event} over the set of binary random variables $\{X_1,\dots,X_c\}$, if there is a subset $S \subseteq [c]$ such that $E = \bigvee_{j\in[c]} X_j$.
\end{definition}


\paragraph{Number of clusters.} For $i > 1$, we have $$\E[c_i] = \sum_{j \in [c_i]} \Prob[X_j^{(i)} = 1] = p\cdot\E[c_{i-1}].$$ Since $c_1 = n$, we have $\E[c_i] = np^{i-1}$.

\paragraph{Last iteration.} For $i = k$, there are $np^{k-1} = n^{1/k}$ clusters in expectation. Since there are $n$ nodes, at most $n^{1+1/k}$ edges are added in expectation for $i=k$.

\paragraph{First $k-1$ iterations.} Consider an alive node $v$ with $d$ adjacent clusters in iteration $i$. Without loss of generality, assume that its $j$-th adjacent cluster is $C^{(i)}_j$ and the smallest edge weight between $v$ and $C^{(i)}_j$ is larger or equal than the smallest edge weight between $v$ and $C^{(i)}_{j-1}$.
Let $X^{(i)}_j$ be the indicator random variable that the cluster $C^{(i)}_j$ is selected in the $i$-th iteration. 
The node $v$ adds at most $$1 + \sum_{j \in [d]} \Prob[\bigvee_{\ell \in [j-1]} X^{(i)}_\ell = 0] < 1 + \sum_{j=0}^{\infty} (1-p)^j = O(1/p)$$ edges in expectation. There are $n$ nodes, so each iteration adds $O(n/p)$ edges. In total and including the last iteration, we get the claimed size bound $O(nk/p) = O(n^{1+1/k}k)$.

\paragraph{Unweighted graphs.} In this case, when a node remains alive during an iteration, it adds at most one edge to the spanner (since we only add edges with \textbf{strictly} smaller weights). So, in total, at most $n(k-1)$ edges are added from nodes that remain alive during an iteration. For the contribution of dead nodes, consider an alive node $v$ with $d$ adjacent clusters in iteration $i$. Let $S \subseteq [c_i]$ with $|S| = d$ be the set of adjacent clusters of $S$. Node $v$ dies in iteration $i$ with probability at most 
$$\Prob[\bigvee_{j \in S} X_j^{(i)} = 0] = (1 - p)^d \leq e^{-pd}.$$ 
So a node can add up to $de^{-pd}$ edges, in expectation. Function $x \rightarrow xe^{-px}$ is maximized at $x = 1/p$ where its value is $\Theta(1/p)$. 
From this and since there are $n$ nodes, each iteration adds $O(n/p)$ edges, in expectation. 
This results in the final expected size of $O(nk/p)$. 
This is the same as in the weighted case. 

To improve this bound, suppose $d$ is larger than a threshold $\tau = \ln(k)/p$. Function $x \rightarrow xe^{-px}$ is decreasing in the range $[1/p,+\infty)$. Thus the total expected contribution of such $v$ in iteration $i$ is $n\cdot(\tau e^{-p\tau}) = O(n\log k/(pk))$. There are $k-1$ iterations, so $O(n\log k/p)$ edges added from these nodes in expectation. The bound for nodes with $d \leq \tau$ is deterministic. The overall contribution of these nodes in all the first $k-1$ iterations is at most $n\tau$ (since each node dies only once and there are at most $n$ nodes). Considering the contribution of $n^{1+1/k}$ edges in the last iteration, the expected size of the final spanner is the claimed bound $O(nk + n^{1+1/k}\log k)$.

\paragraph{High-degree nodes.} For implementation, we need one more ingredient. An alive node $v$ with $d$ adjacent cluster in iteration $i$ is called \textit{high-degree} if $d \geq \xi = 10\ln n/p$. The probability that such a node dies in iteration $i$ is a hitting-event over $\{X^{(i)}_{j}\}_{j \in [c_i]}$ and is at most $(1-p)^{\xi} \leq e^{-p\xi} = n^{-10}.$ So by union bound, no high-degree node dies during iteration $i$ with probability at least $1 - n^{-9}$.

\paragraph{Reducing randomness.} In the above, we assume full independence for sampling in each iteration and between the iterations. The analysis still goes through with less randomness inside each iteration but keeping the independence across the iterations. More concretely, for each $i$, there is a distribution $\mathcal{P}$ from which we can sample with $O(\log n \log \log n)$ random bits to generate $\{X^{(i)}_j\}_{j\in [c_i]}$ such that it approximates each hitting-event with additive factor $1/\poly(n)$. To be more precise, let $E$ be a hitting-event over $\{X^{(i)}_j\}_{j\in [c_i]}$ and let $p_E$ be the probability that $E$ is 1 assuming that each $X^{(i)}_j$s is sampled independently of the other. Let $\tilde{p}_E$ be $\Prob_{\mathcal{P}}[E = 1]$. Then $|p_{E} - \tilde{p}_{E}| \leq n^{-10}$. Let us emphasize that the distribution $\mathcal{P}$ is independent of $E$. For details, please see \cref{app:short-seed-length}. In derandomization, we use the method of conditional expectation where we fix bits of the random seed one by one. So for a polylogarithmic rounds algorithm, such a short seed is needed. 

The only randomized part of each iteration is cluster sampling. If we derandomize this part, the whole algorithm becomes deterministic. The next lemma gives a formal statement of our derandomization for sampling. It is written in a parametric form as we need it later for computing linear size spanners. There is no randomness in the last iteration, so it is already deterministic and we ignore that. A direct implication of the following lemma is \cref{thm:kstretch}.
\begin{lemma}
\label{lem:detbs}
For any positive integer $g$ and real number $p \in (1/n,1)$, there is a deterministic distributed algorithm that ``simulates'' $g$ iterations of Baswana-Sen with sampling probability $p$ in $g^2 \cdot \polylog(n)$ rounds. That is, (1) it adds $O(ng/p)$ and $O(ng + \frac{n \log g}{p})$ edges to the spanner for weighted and unweighted graphs, respectively, (2) the number of clusters in the output of last iteration is at most $np^g$, (3) no high-degree node dies in any of $g$ iterations.
\end{lemma}

\begin{proof}
We derandomize each iteration separately. For iteration $i \in [g]$, we have three objectives:
\begin{enumerate}[(a)]
    \item The number of clusters in the output of iteration $i$, variable $c_{i+1}$, is at most $np^{i}$.
    \item Let $\iota$ be a large enough constant. For weighted graphs, we should add at most $\iota n/p$ edges. For unweighted graphs, nodes that die in the iteration and has at least $\tau = \ln g/p$ adjacent clusters should add at most $(\iota n\log g)/(pg)$ edges. From earlier discussion, the contribution of nodes that remain alive or the one that dies but has $d \leq \tau$ is within the final size budget, deterministically. So we ignore them.
    \item Ensuring that no high-degree node (having more than $10 \ln n/p$ adjacent clusters) dies.
\end{enumerate}
We combine all our objectives into one random variable (also known as utility function). For weighted graphs, we define $U_i^{\mathrm{w}}$ as 
\begin{align}
    \label{eq:w}
U_i^{\mathrm{w}} =  \iota/p^{i+1}\sum_{j \in [c_i]} X^{(i)}_j + \sum_{v \in V_i} (b_v + n^5 h_v).
\end{align}
For unweighted graphs, we define $U_i^{\mathrm{uw}}$ as
\begin{align}
  \label{eq:uw}
  U_i^{\mathrm{uw}} =  (\iota \ln g)/(gp^{i+1})\sum_{j \in [c_i]} X^{(i)}_j + \sum_{v \in V_i} (b_v + n^5 h_v).  
\end{align}

In the definitions of $U_i^{\mathrm{w}}$ and $U_i^{\mathrm{uw}}$, set of alive nodes in iteration $i$ is denoted by $V_i$. Random variable $b_v$ is the number of edges added by node $v$ (in the unweighted case, we set $b_v$ to zero if it is an ignored node). Random variable $h_v$ is one if $v$ is high-degree and dies in iteration $i$. It is zero otherwise. Suppose there is an assignment of $\{X^{(i)}_j\}_{j\in [c_i]}$ that makes $U^{\mathrm{w}}_i$ at most $\iota n/p$. So in this assignment $c_{i+1} = \sum_{j \in [c_i]} X^{(i)}_j \leq np^i$ and $\sum_{v \in V_i} b_v \leq \iota n/p$. Moreover, $\sum_{v \in V_i} h_v = 0$ as $\iota n/p = O(n^2)$ and the summation should be an integer. So all three required conditions (a), (b), and (c) hold. Similarly, for the unweighted case, an assignment with $U^{\mathrm{uw}} \leq (\iota n \log g)(pg)$ suffices. But why does such an assignment exist? For that, suppose we sample each $X^{(i)}_{j}$ independently with probability $p/4$ (and not $p$). From induction, we know that $c_i \leq np^{i-1}$ so $\E[\sum_{j \in [c_i]} X^{(i)}_j] \leq np^i/4$, and using earlier discussion, we know that $\sum_{v \in V_i} \E[b_v] \leq \iota n/(4p)$ and $\sum_{v \in V_i} \E[h_v] \leq n^{-5}$. So, we have: $$\E[U_i^{\mathrm{w}}] \leq \iota n/(4p) + \iota n/(4p) + n^{-5} < \iota n/p.$$ Similarly, we can show that $\E[U_i^{\mathrm{uw}}] < (\iota n \log g)/(pg).$

So there is such a good assignment for $\{X^{(i)}_j\}_{j \in [c_i]}$. Observe that each utility function can be written as $O(n^2)$ hitting-events, so if we approximate full independence with the distribution in \cref{app:short-seed-length}, the above analysis still goes through as it only incurs $O(n^2 \cdot n^{-10}) = O(n^{-8})$ error. To find such an assignment deterministically in $\polylog(n)$ rounds, we use the method of conditional expectation running over a network decomposition. The details are in \cref{app:bs-derand}. 

\paragraph{Computational Aspects.} We do not get work-efficient algorithms in the sense of having only $m\cdot \polylog(n)$ computation. In fact, the amount of computation of each node is slightly super-polynomial $2^{O(\log n \log \log n)} = n^{O(\log \log n)}$. However, this is somewhat similar to (and only better than) recent works that use conditional expectation for derandomization. In~\cite{ghaffari2018derandomizing, deurer2019deterministic}, they use $O(\log n)$-wise independence, and in~\cite{parter2018congested}, they use a distribution with $\log n (\log \log n)^3$ random bits. Hence, they need $n^{O(\log n)}$ and $n^{O((\log \log n)^3)}$ local computations, respectively.
\end{proof}

\subsection{Linear Size Spanners}
We also provide a $\polylog(n)$ rounds derandomization of Pettie's algorithm~\cite{pettie2010distributed} for weighted and unweighted graphs. Our result is stated in \cref{thm:linearsize} and its proofs is deferred to \cref{app:linear-size}. In the following, we define \textit{stretch-friendly} clustering which plays the key role in the stretch analysis of this algorithm and \cref{sec:reduction}. 

\begin{definition}
An $r$-cluster $C$ is stretch-friendly if for any boundary-edge $\{u \not \in C,v \in C\}$ of $C$ with weight $w$, all edges in the unique path from $v$ to the root of $C$ have weight at most $w$. Moreover, for any inside-edge $\{u \in C,v \in C\}$ of $C$ with weight $w$, all edges in the unique path between $u$ and $v$ in $C$ have weight at most $w$. An $r$-clustering ($r$-partition) is stretch-friendly if all of its clusters are stretch-friendly. 
\end{definition}

\begin{observation}
\label{obs:stretch-friendly}
Let $\mathcal{C}$ be a stretch-friendly $r$-partition of $G$ and $H$ be the $r$-cluster-graph induced by $\mathcal{C}$. Denote by $T_C$ the set of edges in the tree of cluster $C \in \mathcal{C}$. The union of an $\alpha$-spanner of $H$ with $\cup_{C \in \mathcal{C}} T_C$ is a $((2r + 1)(\alpha + 1) - 1)$-spanner of $G$.
\end{observation}

%% file: reduction.tex
\section{Deterministic Ultra-Sparse to Sparse Reduction}
\label{sec:reduction}


The main tool of this section is the following lemma.

\begin{lemma}
\label{lem:reduction}
There is a deterministic distributed algorithm that computes a stretch-friendly $O(t)$-partition with at most $n/t$ clusters in $O(t\log^* n)$.
\end{lemma}
To our knowledge, such a result is only known for unweighted graphs, (see Kutten and Peleg~\cite{kutten1998fast}). Note that in the unweighted case, any clustering is stretch-friendly. Round complexity of~\cite{kutten1998fast} is also $O(t \log^* n)$ and our algorithm is arguably simpler.

\begin{remark}
The algorithm of \cref{lem:reduction} can be run in $\polylog(n)$ depth and $m\cdot\polylog(n)$ work in the \pram model.
\end{remark}

Next, we prove \cref{thm:p1-sparse-ultrasparse} and \cref{thm:ultrasparse}.

\begin{proof}[Proof of \cref{thm:p1-sparse-ultrasparse}]
First, we find a stretch-friendly $O(t\cdot s(n))$-clustering $\mathcal{C}$ using \cref{lem:reduction}. Then, using $A$, we compute an $\alpha$-spanner of a $O(t\cdot s(n))$-cluster-graph that is induced by $\mathcal{C}$. This spanner along with trees corresponding to the clusters in $\mathcal{C}$ is a $O(t\cdot s(n)\alpha)$-spanner of $G$ according to \cref{obs:stretch-friendly}. The spanner has $n + O(n/t)$ edges since at most $n-1$ edges are in the union of trees of $\mathcal{C}$ and algorithm $A$ adds $s(n/(t\cdot s(n))) \frac{n}{t\cdot s(n)} = O(n/t)$ edges. Multiplying $t$ by a large enough constant gives a spanner with $O(t\cdot s(n)\alpha)$ stretch and $n + n/t$ edges. The round complexity follows from the round complexity of $A$ and \cref{lem:reduction}.
\end{proof}

\begin{proof}[Proof of \cref{thm:ultrasparse}]
We apply \cref{thm:p1-sparse-ultrasparse} with algorithm $A$ being the linear size algorithm of \cref{thm:linearsize}. For this $A$, the output has size $O(N)$, so $s(N) = O(1)$. For function $T(N,r)$, recall that $A$ is split into $O(\log^* n)$ phases. The input of each phase is a cluster-graph. It is discussed in the proof of \cref{thm:linearsize} that the radius of input for each phase is linearly multiplied in its round complexity as it stretches the dilation by $r$. So $T(N,r) = r\cdot \polylog(N)$ which completes the proof.
\end{proof}

The algorithm of \cref{lem:reduction} is described in the following. Its proof is deferred to \cref{app:stretch-friendly}. We gradually construct the partition in $\lceil\log t\rceil$ iterations. The output of iteration $i$ (input of iteration $i+1$) is a stretch-friendly $(3\cdot 2^i - 1)$-partition with each cluster has size at least $2^i$. The input of first iteration is the trivial partition (one cluster for each node). Details of iteration $i$ with input partition $\mathcal{C}$ are as follows:

\begin{enumerate}[(1)]
    \item Each cluster $C \in \mathcal{C}$ computes it size.
    \item Each cluster $C$ finds its minimum weight boundary-edge (breaking ties arbitrarily) and orients it out from $C$.
    \item A 3-coloring of the cluster-graph induced by $\mathcal{C}$ considering only oriented edges is computed.
    \item A maximal matching between \textit{small} clusters is computed. A cluster is \textit{small} if its size is less than $2^i$ and is \textit{large} otherwise.
    \item A partition $\mathcal{C'}$ (input of iteration $i+1$) is created: First, we merge matched clusters and put one cluster for each merged cluster in $\mathcal{C'}$. Then, each large cluster is added to $\mathcal{C'}$. In the end, each unmatched cluster is merged to the cluster of its outgoing neighbor in $\mathcal{C'}$ (this neighbor is in $\mathcal{C'}$ as it is either a matched small cluster or a large cluster). When we merge two clusters $C_1$ and $C_2$ with an edge oriented from $C_1$ to $C_2$, the root of the new cluster is the root of $C_2$. 
\end{enumerate}

%% file: P2-300_sparse_to_low_diameter.tex
\section{Deterministic Unweighted Spanner via Low-Diameter Clusterings}
\label{sec:unweighted_spanner}

In this section, we sketch how one can efficiently compute ultra-sparse spanners with $n + n/t$ edges by computing $O(\log n)$ $t'$-separated strong diameter clusterings for $t' = O(t \log n)$. 
\begin{definition}[$t$-separated Strong Diameter Clustering]
\label{def:low_diameter_clustering}
Assume we are given an unweighted graph $G$ and a parameter $t > 0$. 
A $t$-separated low diameter clustering with strong diameter $D$ is a clustering $\fC$ such that:
\begin{enumerate}
    \item For each $C \in \fC$ we have $\diam(C) \le D$.
    \item For each $C_1 \not= C_2 \in \fC$ we have $d(C_1, C_2) \ge t$. 
    \item We have $|\bigcup_{C \in \fC} C| \ge n/2$. 
\end{enumerate}
\end{definition}

The high-level idea for the spanner construction is to compute a low-diameter clustering covering \emph{all} the nodes such that the total number of neighboring clusters is at most $n/t$. Then, adding for each cluster a low-diameter spanning tree to the spanner and for each of the at most $n/t$ neighboring clusters an arbitrary edge between the two clusters results in a spanner with $n + n/t$ edges and the stretch being on the order of the diameter of each cluster.

Note that obtaining a clustering that covers all the nodes is simple: we iteratively compute a low-diameter clustering of the yet unclustered nodes. In each iteration, the number of unclustered nodes decreases by a factor of two, and therefore each node is clustered after $O(\log n)$ iterations.

To achieve a small number of neighboring clusters, we can do the following in each of the $O(\log n)$ iterations. 
Start with a clustering with separation $100t \log n$. Now, repeatedly grow a given cluster by adding its neighbors to it until it would grow less than by a multiplicative factor of $1+1/t$.
The growth stops after at most $10t\log n$ steps as $(1+1/t)^{10t\log n } > n$. 
The clusters still remain well-separated and the property that each final cluster $C$ neighbors with at most $|C|/t$ nodes implies that the cluster is  ``responsible'' for at most $|C|/t$ neighboring clusters in the final clustering. This means that in the end there are at most $n/t$ neighboring pairs of clusters, as desired.

A formal statement of this reduction together with its proof can be found in \cref{sec:low_diameter_clusterings}. In fact, the proof strengthens the idea presented above and shows that it suffices that the clusters have a separation of $100t$ instead of $100t \log n$.

In \cref{sec:low_diameter_clusterings}, we extend the result by showing how to compute spanners with $\tilde{O}(n)$ edges from so-called weak-diameter clusterings, a more relaxed notion compared to strong-diameter clusterings, for which more efficient \congest algorithms exist \cite{rozhovn2020polylogarithmic, ghaffari_grunau_rozhon2020improved_network_decomposition} (compared to strong-diameter clusterings \cite{chang2021strong}).
We conclude \cref{sec:low_diameter_clusterings} by showing that, using the reduction from ultra-sparse spanners to sparse spanners and the folklore reduction from weighted to unweighted spanners, we can efficiently compute weighted ultra-sparse spanners.

%% file: appendix.tex
\section{Some Other Related Work}
\label{app:otherRelated}
\paragraph{Centralized and parallel algorithms for ultra-sparse weighted spanners} We next discuss algorithms from the centralized and parallel computation settings that compute ultra-sparse weighted spanners, and elaborate why they are not applicable in the distributed setting. 

One (centralized) approach to construct ultra-sparse spanners with good stretch for weighted graphs is to give an efficient implementation for the greedy algorithm of \cite{althofer1993sparse}. This was done by Roditty and Zwick \cite{RodittyZ04} who showed that the greedy algorithm can be implemented in $\tilde{O}(n^2)$ time in the centralized model of computation. Unfortunately, this implementation is inherently sequential and not amenable to extend to work-efficient distributed or parallel implementations. Dubhashi et al. \cite{DMPRS05} devised a distributed implementation of the greedy algorithm in the \local model of distributed computation, but this algorithm only works for unweighted graphs and is not work-efficient~\footnote{The \local model allows sending of arbitrarily large messages and performing arbitrarily heavy local computations. The algorithm of \cite{DMPRS05} indeed employs such heavy local computations.}. 

In terms of work-efficient algorithms, to the best of our knowledge, there are currently only two parallel algorithms for building ultra-sparse sub-graphs with distance-approximation guarantees for weighted graphs, namely Blelloch et al.~\cite{BGKMPT14} and Li~\cite{Li20}. Both algorithms seem inherently randomized and not suitable for distributed computational models: 

The algorithm of Blelloch et al.~\cite{BGKMPT14} provides ultra-sparse subgraphs with a guarantee on the {\em average} stretch of edges of the original graph, as opposed to the worst-case stretch guarantee of spanners. While average-stretch is good enough for some spectral-sparsification algorithms, it is insufficient for many other applications, e.g., for Li's recursive scheme \cite{Li20} for computing $(1+\eps)$-transshipment and $(1+\eps)$-approximate single-source shortest paths.

The second randomized parallel algorithm, due to Li \cite{Li20}, provides for any $t \ge 1$ a spanners with $n + n/t$ edges and $O(t^2 \cdot \log^2 n \cdot \log^2 \log n)$ stretch using $O(t \cdot \log^3 n \cdot \log \log n \log^* n)$ depth and $O(m \log n)$ work, while succeeding with constant probability.\footnote{When run on an $n$-vertex $m$-edge graph with a parameter $k$, the algorithm of \cite{Li20}, with constant probability, produces an $O(k^2)$-spanner with $n + O({{m \cdot \log n} \over k})$ edges in $O(k \log^2 n \log^* n)$ time and $O(m \log^2 n)$ work. By first building an $O(\log n)$-spanner with $m = O(n \log\log n)$ edges via the algorithm of \cite{miller2015improved}, and then running the algorithm of \cite{Li20} on top of this spanner with $k = t \log n \log\log n$, we obtain the bounds cited here.
} Because the approach of \cite{miller2015improved} for weighted graphs, which is also utilized by Li~\cite{Li20}, requires $\Omega(\frac{\log n}{\log \log n})$ rounds of contractions, these two algorithms cannot be implemented distributedly. 


\section{Short Seed Length}
\label{app:short-seed-length}
We want to approximate a hitting-event $E$ over $n$ binary random variables $X_1, \dots, X_n$ where each $X_i$ is sampled independently with probability $p$. For this, we can use the family of hash functions provided in the work of Gopalan and Yehudayoff~\cite{gopalan2020concentration}.

\begin{theorem}[\cite{gopalan2020concentration}, Theorem 1.9]
For positive integers $N$ and $M$, and for error parameter $\delta > 0$, there is a family of hash functions $\mathcal{H}$ from $[N]$ to $[M]$ with size $2^r$ where $r = O((\log \log N + \log (M/\delta))\log \log (M/\delta))$, such that for every subsets $T_1, \dots, T_N \subseteq [M]$, we have:
\begin{equation*}
    \left | \Prob_{h\sim \mathcal{H}}[\forall i \in [N], h(i) \in T_i] -
            \Prob_{u\sim \mathcal{U}}[\forall i \in [N], u(i) \in T_i]
    \right| \leq \delta
\end{equation*}
where $\mathcal{U}$ is the family of all functions from $[N]$ to $[M]$. We can sample from $\mathcal{H}$ in polynomial time.
\end{theorem}

Set $N$ to $n$, $M$ to $1/p$, and $\delta$ to $n^{-{10}}$. So the seed length is $O(\log n \log \log n)$ as in our applications, we always have $p \geq 1/\poly(n)$. To sample $\{X_i\}_{i \in [n]}$, we sample a random hash function from $\mathcal{H}$, and then set $X_i$ to $1$ if $h(i)$ is $1$ and set it to zero if $h(i) \neq 1$. This approximates any hitting event $E$ with error $\delta$. To show this, suppose $S \subseteq [n]$ is the corresponding subset of $E$, that is $$\Prob[E = 0] = \Prob[\bigwedge_{i\in S} X_i = 0] = (1 - p)^{|S|}$$ For $i \in [N]$, we set $T_i$ to $[M] \setminus \{1\}$ if $i \in S$ and we set $T_i$ to $[M]$ if $i \not \in S$. Note that:
\begin{align*}
    \Prob_{u\sim \mathcal{U}}[\forall i \in [N], u(i) \in T_i] &= 
    \Prob_{u\sim \mathcal{U}}[\forall i \in S, u(i) \in T_i] \\ 
    & = \Prob_{u\sim \mathcal{U}}[\forall i \in S, u(i) \neq 1]\\
    & = (1 - 1/M)^{|S|} = (1 - p)^{|S|}
\end{align*}
On the other hand, we have:
\begin{align*}
    \Prob_{h\sim \mathcal{H}}[\forall i \in [N], h(i) \in T_i] &=
    \Prob_{h\sim \mathcal{H}}[\forall i \in S, h(i) \in T_i] \\
    &= \Prob_{h\sim \mathcal{H}}[\forall i \in S, X_i \neq 1]\\
    &= \Prob_{h\sim \mathcal{H}}[\bigwedge_{i\in S} X_i = 0]
\end{align*}
So with this setting, we can capture hitting-event $E$. And so the distribution that is provided by $\mathcal{H}$ gives a $\delta=n^{-10}$ approximation of this event with only $O(\log n \log \log n)$ random bits.

\begin{remark}
Note that $p$ might not be in the form of $1/M$. In that case, assuming $p = \Omega(1/\poly(n))$, we can write $p = M'/M - 1/\poly(n)$ with $M$ and $M'$ being $\poly(n)$ bounded integers. If in the above, we set $T_i$ to $[M] \setminus [M']$ for $i \in S$ and $T_i$ to $[M]$ for $i \not \in S$, we again get approximation with polynomially small error with $O(\log n \log \log n)$ seed length.
\end{remark}


\section{Baswana-Sen Derandomization}
\label{app:bs-derand}

We deterministically assign sample/unsample to clusters using the method of conditional expectation running on network decomposition. But first, let us recall the definition of network decomposition.

A $(Q,D)$ weak-diameter network decomposition of an unweighted graph $G$ is a partitioning of nodes of $G$ into clusters, colored with $Q$ colors, such that nodes of each cluster are at most at distance $D$ in $G$. Moreover, it guarantees that there is no edge between clusters with the same color. According to the work of Rozho{\v{n}} and Ghaffari~\cite{rozhovn2020polylogarithmic}, one can find a weak-diameter network decomposition in $\polylog(n)$ rounds deterministically. Indeed, they show that we can find a network decomposition of any constant power of $G$ in $\polylog(n)$ rounds. The $t$-th power of $G = (V, E)$, denoted by $G^t$, is a graph on $V$ with an edge between any two nodes with distance at most $t$ in $G$. For our derandomization, we need network decomposition of $G^2$. Concretely, it means that any two clusters $C$ and $C'$ with the same color have a distance of at least three in $G$.

\begin{theorem}[\cite{rozhovn2020polylogarithmic}, Theorem 2.12]
\label{thm:networkdecomposition}
There is a deterministic distributed algorithm that computes a $(\log n, \polylog(n))$ weak-diameter network decomposition of $G^2$.

Moreover, for each color class and each cluster $C$ on this class, a Steiner tree $T_\mathcal{C}$ with radius $\polylog(n)$ is associated. Terminal nodes of $T_C$ are the set of nodes of $C$ and each edge of $G$ appears in $O(\polylog(n))$ of these Steiner trees.
\end{theorem}

The second part of the theorem above guarantees that for any color $q$, we can simulate convergecast and broadcast on the collections of clusters with color $q$ simultaneously as each cluster can communicate through its Steiner tree. Since each edge is in only few Steiner trees, the congestion overhead of simultaneous communication cannot be more than $\polylog(n)$ rounds.

\begin{proof}[Proof of \cref{lem:detbs}]
Let $\mathcal{C}$ be the input clustering of iteration $i$ (an $(i-1)$-partition of nodes that are alive at the beginning of iteration $i$) and $H$ be the cluster-graph induced by $\mathcal{C}$. First, we find a weak-diameter network decomposition of $H^{2}$ into $O(\log n)$ color classes such that each cluster has weak diameter $\polylog(n)$. For that, We can run the $\polylog(n)$ rounds algorithm of~\cite{rozhovn2020polylogarithmic} on $H$. The algorithm can be easily adapted to run on cluster-graphs as its operation only needs broadcast and convergecast, and moreover, it can run on a network where each node has $\polylog(n)$ bits of memory. Since each node of $H$ corresponds to a tree with radius $i-1$ of $G$, the total running time is $O(i \cdot \polylog(n))$ rounds. To avoid confusion, we reserve the term ``cluster'' for referring to the clusters in Baswana-Sen algorithm and we refer to each cluster of network decomposition by ``ND-cluster''.

The contribution of each node to both $U^{\mathrm{w}}_i$ and $U^{\mathrm{uw}}_i$ (see \cref{eq:w,eq:uw}) can be written as the contribution of each individual alive node of $G$.  The contribution of $X^{(i)}_j$ (indicator random variable that determines whether $j$-th cluster in iteration $i$ is sampled or not) is considered as part of a contribution of the root of the $j$-cluster. Moreover, observe that each node of $G$ can compute its contribution by looking at its 1-hop neighborhood. In the following, we show how we can derandomize a utility function with these properties.

We go over the color classes of network decomposition one by one. Suppose we assign sample/unsample labels for each Baswana-Sen cluster that is in a ND-cluster with color 1 to $q-1$ and now we are working on color $q$. Suppose all ND-clusters with color $q$, independent of each other, sample the Baswana-Sen clusters that are inside them with $O(\log n \log \log n)$ bits from the distribution in \cref{app:short-seed-length}\footnote{Note that if two distributions approximate the independent distribution, their product is also an approximation.}. In order to derandomize sampling of each ND-cluster, we go over $O(\log n \log \log n)$ bits of its random seed one by one and fix them in such a way that the expected value of the utility function condition on those fixed bits is at most be the expected value of utility function without fixing those bits.

Consider one ND-cluster $\mathcal{D}$ with color $q$. Let $V_{\mathcal{D}}$ be the set of alive nodes of $G$ in $\mathcal{D}$ and let $N(\mathcal{D})$ be the union of $V_{\mathcal{D}}$ and all neighbors of $V_{\mathcal{D}}$ in $G$. Observe that the randomness of $\mathcal{D}$ only affects the contribution of nodes in $N(\mathcal{D})$. Since we do a network decomposition for the second power of $H$, we have two properties: (a) all nodes of one Baswana-Sen cluster is in one ND-cluster, and (b) for any two ND-clusters $\mathcal{D}$ and $\mathcal{D'}$, the set of affected nodes by the randomness of $\mathcal{D}$ and $\mathcal{D'}$ are disjoint, i.e. $N(\mathcal{D}) \cap N(\mathcal{D'}) = \emptyset$. Here is how we fix the first random bit for $\mathcal{D}$. Other bits can be fixed similarly. The first bit is fixed to $b \in \{0,1\}$ if
\begin{equation*}
\E[\text{contribution of } N(\mathcal{D}) \mid \text{first bit is } b] \leq
\E[\text{contribution of } N(\mathcal{D}) \mid \text{first bit is } 1-b]
\end{equation*}
To decide for which $b$ the conditional expectation is smaller, each node in $N(\mathcal{D})$ computes two values: its contribution to the utility function if the first random bit equals zero/one. Each of those two values can be computed by taking a sum over all assignments of bits that are not yet fixed. So it may take $2^{O(\log n \log \log n)} = n^{O(\log \log n)}$ time. Then for each Baswana-Sen cluster $C$ in $\mathcal{D}$, all nodes in $C$ aggregates the sum of their contributions in the root of the cluster. This takes $O(i)$ rounds\footnote{We trim the computed values into $O(\log n)$ bits. This is safe as it only incurs $1/\poly(n)$ error.}. Each cluster $C$ in $\mathcal{D}$ is indeed represents a node in $H$. Now, these nodes in $H$ aggregates their computed values into the leader of the ND-cluster $\mathcal{D}$. This can be done in $i\cdot \polylog(n)$ rounds via the Steiner trees of $\mathcal{D}$ (as it is mentioned in \cref{thm:networkdecomposition}). The leader then decides which bit should be fixed and broadcast it to all nodes in $N(\mathcal{D})$. Since fixing one bit takes $i\cdot \polylog(n)$ rounds and as there are $O(\log n \log \log n)$ bits, fixing all of them also takes $i \cdot \polylog(n)$ rounds. There are $O(\log n)$ color classes in the network decomposition, so simulating the sampling for iteration $i$ of Baswana-Sen takes $i \cdot \polylog(n)$ rounds. There are $g$ iterations, so the final round complexity is $\sum_{i=1}^{g} i\cdot \polylog(n) = g^2 \cdot \polylog(n)$.
\end{proof}


\section{Linear Size Spanners}
\label{app:linear-size}
\begin{proof}[Proof of \cref{thm:linearsize}]
The algorithm consists of $P = O(\log^* n)$ phases. Phase $i$ is given an input graph $G_i$ with $n_i$ nodes. Using the algorithm of \cref{lem:detbs}, we run the Baswana-Sen algorithm for $g_i = (1 + I_{w})x_i(1 + 2\frac{\log \log x_i}{\log x_i})$ iterations and with sampling probability $\frac{1}{x_i}$. Here, $\mathrm{I}_w$ is 1 if we are in the weighted case and 0 if we are in the unweighted case. This is why there is an extra $2^{\log^* n}$ factor in the stretch bound of the spanner for the weighted case. Parameter $x_i$ is determined later. At the end of phase $i$, we define graph $G_{i+1}$ (input of the next phase) as the $g_i$-cluster-graph that is induced by $g_i$-clustering of running $g_i$ iterations of Baswana-Sen. The input to the first phase is the network graph $G$. We show that the there is a sequence $x_1, \dots, x_{P}$ such that the edges that are added to the spanner in this $P$ phases gives us a spanner with the claimed size and stretch.

The goal of phase $i$ is to construct a spanner for $G_i$ with stretch $s_i = \prod_{j=i}^{P}(2g_j + 1)$. Note that since $G_1$ is $G$, proving this implies that our output is a $s_1$-spanner of $G$. In phase $i$, after running the sampling for $g_i$ iterations, all dead edges of $G_i$ has stretch $2g_i - 1 \leq s_i$ according to \cref{lem:bs-stretch}. So we do not need to worry about them in terms of stretch and we can safely remove them from further consideration. Again, according to \cref{lem:bs-stretch}, observe that the $g_i$-clustering that we have after Baswana-Sen iterations is stretch-friendly. From \cref{obs:stretch-friendly}, if we can find an $s_{i+1}$-spanner for $G_{i+1}$, we get a spanner of stretch $s_i$ for $G_i$. So if we show that the last phase gives a spanner with stretch $s_P$, the stretch of other phases follows automatically by induction. For that, we design the sequence $x_1, \dots, x_P$ such that all nodes of $G_P$ are dead at the end of the last phase and as a result, all edges of $G_p$ are dead which means that their stretch in the spanner is bounded by $2g_P - 1 < s_P$.

Let $P$ be the largest integer such that $\log^{(P)}(n) \geq \alpha_0$. Here, $\log^{(P)} n$ represents the $P$-th iteration of $\log$ and $\alpha_0$ is the constant in \cref{lem:x-seq}. Observe that $P \leq \log^* n$ and that $\log^{(P)} n$ is a constant. We define the sequence of $x_i$ as follows:
\begin{equation*}
    x_1 =  \frac{\log^{(P)} n}{\log^{(P + 1)} n}, x_2 = \frac{\log^{(P-1)} n}{\log^{(P)} n}, \dots, x_P = \frac{\log n}{\log \log n}
\end{equation*}
Note that:
\begin{align*}
    s_1 = \prod_{i=1}^{P}(2g_i + 1)
    &\leq 2^{P} \left(\prod_{i=1}^{P} g_i\right) e^{\sum_{i=1}^{P} \frac{1}{2g_i}}\\
    &\leq (2(1 + I_w))^P \left(\prod_{i=1}^{P} x_i \right)e^{\sum_{i=1}^{P} \frac{2\log \log x_i}{\log x_i} + \frac{1}{2g_i}}\\
    &= (2(1 + I_w))^{P}\cdot (\log n) \cdot O(1).
\end{align*}
We use the inequality $1+x \leq e^x$ and the fact that the summations $\sum_{i=1}^{P}\frac{1}{g_i}$ and $\sum_{i=1}^{P}\frac{2\log \log x_i}{\log x_i}$ are constants. This is because the sequence of $x_i$ (and so $g_i$) are exponentially growing and so the summations are dominated by their first term. So the output spanner has stretch $O(\log n \cdot 2^{\log^* n})$ for unweighted graphs and $O(\log n \cdot 4^{\log^* n})$ for weighted graphs. There is one remaining ingredient though. We have to show that in last phase, all nodes of $G_P$ die. Note that at the end of this phase, we have at most $n_P/x_P^{g_P}$ clusters according to \cref{lem:detbs}. On the other hand, from \cref{lem:x-seq}, we have $x_P^{g_P} \geq n \geq n_P$ (set $\alpha$ to $n$ in the lemma). So there is no cluster at the end of last phase and so all nodes are dead. This concludes the proof for the claimed stretch.

To bound the spanner size, let us first consider the unweighted case. Note that according to \cref{lem:detbs}, we add $O(n_i g_i + n_i x_i \log g_i) = O(n_i x_i \log x_i)$ edges to the spanner in phase $i$. From \cref{lem:x-seq}, $x_i \log x_i \leq x_{i-1}^{g_{i-1}}$ (set $\alpha$ to $\log^{(P-i+1)} n$ in the lemma). On the other hand, recall that $n_i \leq n_{i-1}/x_{i-1}^{g_{i-1}}$. So for any $i > 1$, the contribution of phase $i$ is $O(n_{i-1})$. In the first phase, $x_1$ is a constant so $n_1 x_1 \log g_1 = O(n_1)$. So size of the spanner is bounded by $O(n_1 + \sum_{i=2}^{P} n_{i-1})$. Again, using $n_i \leq n_{i-1}/x_{i-1}^{g_{i-1}}$, we can write $\sum_{i=2}^{t} n_{i-1} = O(n_1)$ as the summation is dominated by its first term. So in total, the size of the spanner is $O(n_1) = O(n)$. For the weighted case, note that in phase $i$ we add $O(n_i g_i x_i) = O(n_i x_i^2)$ edges. So $x_i \log x_i$ in the unweighted case, is replaced by $x_i^2$. On the other hand, we double the number of iterations of each phase. So similar to the unweighted case, we can show that phase $i$ in the weighted case is also adds $O(n_{i-1})$ edges and so the final spanner size is $O(n)$.

To implement the algorithm, the main challenge is that how we can run \cref{lem:detbs} on a cluster-graph. For each node $v$ in $G_i$, there is a tree between nodes of $\inv^{G_{i}\rightarrow G_{i-1}}(v)$ with radius at most $g_{i-1}$. Subsequently, there is a tree among $\inv^{G_{i} \rightarrow G}(v)$ with radius at most $\frac{s_1}{s_i} = \prod_{j=1}^{i-1}(2g_j + 1)$. This rooted tree can be constructed in $O(\frac{s_1}{s_i}) = O(\log n)$ rounds. The root of this tree is responsible for the node $v$ when we run phase $i$. Recall that in phase $i$, we run Baswana-Sen for $g_i$ iterations with sampling probability $\frac{1}{x_i}$. To run that, each node of $G_i$ should know its adjacent clusters. If a node has $d$ adjacent clusters, then the root can discover this $d$ neighbors in $O(d \cdot \frac{s_1}{s_i})$ rounds. This can be problematic if $d$ is large. However, note that according to \cref{lem:detbs} all nodes with $\Omega(x_i \log n)$ adjacent clusters remain alive in each iteration. So a node only adds at most $\min(\Theta(x_i \log n), d) = O(\log^2 n)$ edges in each iteration. So it should only know the first $O(\log^2 n)$ clusters (when we order the $d$ adjacent clusters of node $v$ according to the smallest weight of an edge between them and $v$). So $v$ can find these clusters in $\polylog(n)$ rounds. The other ingredient of derandomization in \cref{lem:detbs} is running the algorithm of \cref{thm:networkdecomposition} for network decomposition on cluster-graphs. But recall that we already discussed do this in the proof of \cref{lem:detbs}. This concludes that the round complexity of each phase and as a result the whole algorithm (since there are only $O(\log^* n)$ phases) is $\polylog(n)$.
\end{proof}

\begin{lemma}
\label{lem:x-seq}
There is a sufficiently large constant $\alpha_0$, such that for any $\alpha > \alpha_0$, we have 
\begin{equation}
\label{eq:x-seq-claim}
x \log x \leq \alpha \leq y^z
\end{equation}
where $x = \frac{\alpha}{\log \alpha}$, $y = \frac{\log \alpha}{\log \log \alpha}$, and $z = y(1 + \frac{2\log \log y}{\log y})$.
\end{lemma}
\begin{proof}
For the ease of notation, we first take a log from \eqref{eq:x-seq-claim}:
\begin{equation}
\label{eq:x-seq-1}
    \log x + \log \log x \leq \log \alpha \leq z\log y
\end{equation}
Since $\log x = \log \alpha - \log \log \alpha$ and $\log \log x \leq \log \log \alpha$, we have $x \log x \leq \log \alpha$. Next, we show that $\log \alpha \leq z\log y$. Note that:
\begin{equation}
\label{eq:x-seq-2}
    z\log y = y\log y + 2y \log \log y
\end{equation}
Replacing $\log y$ with $\frac{\log \alpha}{\log \log \alpha}$ in the first term, we get 
\begin{equation}
\label{eq:x-seq-3}
    y\log y = \log \alpha - \log \alpha \cdot \frac{\log \log \log \alpha}{\log \log \alpha} = \log \alpha - y \log \log \log \alpha
\end{equation}
If we show that $2\log \log y \geq \log \log \log \alpha$, then the second summand in the RHS of \eqref{eq:x-seq-2} is larger than the negative term of \eqref{eq:x-seq-3} and we are done. For this, note that
\begin{equation*}
    \log \log y = \log (\log \log \alpha - \log \log \log \alpha)
\end{equation*}
Since $\alpha$ is sufficiently large, we have $\log \log \log \alpha \leq \frac{\log \log \alpha}{2}$. So $\log \log y \geq \log \log \log \alpha - 1$. This implies:
\begin{equation*}
    2\log \log y \geq 2(\log \log \log \alpha - 1) \geq \log \log \log \alpha
\end{equation*}
and concludes the proof.
\end{proof}


\section{Stretch-Friendly Clustering}
\label{app:stretch-friendly}
\begin{proof}[Proof of \cref{lem:reduction}]
Consider a cluster $C' \in \mathcal{C'}$. We have to prove three properties for $C'$: (a) Size of $C'$ is at least $2^i$, (b) Radius of $C'$ is less than $3\cdot 2^{i}$, and (c) $C'$ is stretch-friendly. Let us first show that $C'$ satisfies (a) and (b). For this, we consider two cases based on how $C'$ is created:
\begin{itemize}[-]
    \item $C'$ is consisting of a big cluster $C$ with a set of small clusters that has an outgoing edge to $C$. The size of $C'$ is at least $2^i$ as it contains $C$. Each small cluster has size at most $2^i - 1$ (so their hop-diameter is bounded by $2^i - 2$). According to the definition of merging two clusters, the root of $C'$ is the same as the root of $C$. Since the radius of $C$ is at most $3\cdot 2^{i-1} - 1$, the final radius of $C'$ is at most $(2^{i} - 2 + 1) + (3\cdot 2^{i-1} - 1) < 3 \cdot 2^i$.

    \item $C'$ is consisting of two matched small clusters $C_1$ and $C_2$ with a set of small clusters that has an outgoing edge to $C_1$ and a set of small clusters that has an outgoing edge to $C_2$. Without loss of generality, suppose the matched edge between $C_1$ and $C_2$ is oriented from $C_1$ to $C_2$. So the root of $C'$ is the root of $C_2$. There are at least two small clusters in $C'$, so its size is at least $2\cdot 2^{i-1} = 2^i$. As we have discussed before, each small cluster has hop-diameter $2^i - 2$. So the radius of $C'$ is at most $2 \cdot (2^i - 2 + 1) + (2^i - 2) < 3 \cdot 2^i$.
\end{itemize}

It only remains to show the property (c), i.e. $C'$ is stretch-friendly. Note that each cluster in $\mathcal{C'}$ is constructed from merging several clusters in $\mathcal{C}$. And we know by induction that $\mathcal{C}$ is a stretch-friendly partition. So if we show that two stretch-friendly clusters are merged to a stretch-friendly cluster, we are done.

Consider two stretch-friendly clusters $C_1$ and $C_2$ are merged to a cluster $C_{12}$. Suppose the merge edge $\{u_1, u_2\}$ where $u_1 \in C_1$ and $u_2 \in C_2$ is oriented from $C_1$ to $C_2$ and has weight $w$. Moreover, let $\mathrm{root}_1$ and $\mathrm{root}_2$ be the roots of $C_1$ and $C_2$, respectively. From the definition of merge, the root of $C_{12}$ is $\mathrm{root}_2$. Consider a boundary edge with weight $w'$ of $C_{12}$ that is incident to a node $v \in C_{12}$. If $v \in C_1$, then $w' \geq w$. Each edge in the path from $v$ to $\mathrm{root}_1$ has weight at most $w'$ as $C_1$ is stretch-friendly. Each edge in the path from $\mathrm{root}_1$ to $u_1$ has weight at most $w$ (again, because $C_1$ is stretch-friendly) and each edge in the path from $u_2$ to $\mathrm{root}_2$ has weight at most $w$ (since $C_2$ is stretch-friendly). So since $w' \geq w$, all the edges in the path from $v$ to the root of $C_{12}$ which is $\mathrm{root}_{2}$ have weight at most $w'$. The other is that $v \in C_2$. This case is trivial as $\mathrm{root}_{2}$ is root of both $C_2$ and $C_{12}$. So for all boundary edges, $C'$ is stretch-friendly. It remains to show the stretch-friendly property for all inside-edges of $C_{12}$. For inside edges of $C_1$ and $C_2$, we already know that the property holds. So suppose an inside-edge with weight $w'$ between $v_1 \in C_1$ and $v_2 \in C_2$. Exact similar argument, shows that all edges in the path from $v_1$ to $\mathrm{root}_2$ has weight at most $w'$ and all edges in the path from $v_2$ to $\mathrm{root}_2$ has weight at most $w'$ as well. So all the edges in the path from $v_1$ to $v_2$ has weight at most $w'$. This completes the proof that $C_{12}$ is stretch-friendly and as a result shows that $\mathcal{C'}$ (the output of iteration $i$) is stretch-friendly.

Since there are $\lceil \log t\rceil$ iterations, the output of the last iteration is a stretch-friendly $O(t)$-partition where each cluster has size at least $t$. So there are at most $n/t$ clusters.  

The only remaining piece of proving \cref{lem:reduction} is proving the claimed round complexity. For that, notice that steps (1), (2), (4), and (5) of the algorithm can be easily run in $O(2^i)$ rounds since the hop-diameter of each cluster is $O(2^i)$. Step (3) needs $O(2^i \log^* n)$ rounds as it is well-known that a graph with out-degree one can be 3-colored in $O(\log^* n)$ rounds~\cite{Linial87}. And in our case, each node of a graph is a cluster with hop-diameter $O(2^i)$. So the total number of rounds of all the $\lceil \log t\rceil$ iterations is $O(\sum_{i=1}^{\log t} 2^i \log^* n) = O(t\log^* n)$.
\end{proof}

\section{Spanners via Low-Diameter Clusterings}
\label{sec:low_diameter_clusterings}

In \cref{sec:strong_implies_ultra_sparse}, we show that an efficient algorithm for computing a well-separated strong-diameter clustering implies an efficient algorithm for unweighted ultrasparse spanners.
In \cref{sec:weak_implies_sparse}, we show that an efficient algorithm for computing a $3$-separated weak-diameter clustering implies an efficient algorithm for unweighted sparse spanners.
Finally, in \cref{sec:putting_everything_together}, we combine the result from the previous section together with (1) the weak-diameter clustering algorithm of \cite{rozhovn2020polylogarithmic}, (2) the folklore reduction from weighted to unweighted spanners, and (3) the reduction from ultra-sparse spanners to sparse spanners to show that we can efficiently compute a weighted ultra-sparse spanner.
\subsection{Unweighted Ultra-Sparse Spanners via Strong-Diameter Clusterings}
\label{sec:strong_implies_ultra_sparse}
In this section, we show that unweighted ultra-sparse spanners can be efficiently computed given an efficient algorithm for computing well-separated strong-diameter clusterings.

\begin{restatable}{theorem}{unweighted}
\label{thm:unweighted_spanner_warmup}
Let $ t > 0$. 
Assume that for any $t'$ one can construct a $t'$-separated strong-diameter clustering with diameter $\Dldc(n, t')$ in $\Tldc(n, t')$ rounds of the deterministic (randomized) \congest model. 

Then, there is a deterministic (randomized) distributed algorithm that, given an unweighted graph $G$, constructs an $\alpha$-spanner $H$ such that
\begin{enumerate}
    \item $|E(H)| \le n + n/ t$,
    \item $\alpha = O(\Dldc(n, 10 t) + t)$ 
\end{enumerate}
in $O(\log n) \cdot \left( \Tldc(n, 10 t) + \Dldc(n,10 t) + t \right)$ \congest rounds. 
\end{restatable}

The proof follows along the lines of the proof sketch given in \cref{sec:unweighted_spanner}, with two differences.

\begin{enumerate}
    \item We start with $\Theta(t)$-separated clusters instead of $\Theta(t\log n)$-separated ones. 
    This is because it is enough for our purposes that a large fraction of clusters needs to stop growing after $O(t)$ steps; we do not need to wait for $O(t \log n)$ steps after all of them finish. 
    \item Choosing an arbitrary edge between any two neighboring clusters of a given clustering is a non-trivial task. Therefore, we gradually compute a set of edges $E^{inter}$ such that the size of $E^{inter}$ is sufficiently small, and between any two neighboring clusters there exists at least one edge in $E^{inter}$ connecting these two clusters.
\end{enumerate}

In particular, we obtain the following lemma.

\begin{figure*}
    \centering
    \includegraphics[width = \textwidth]{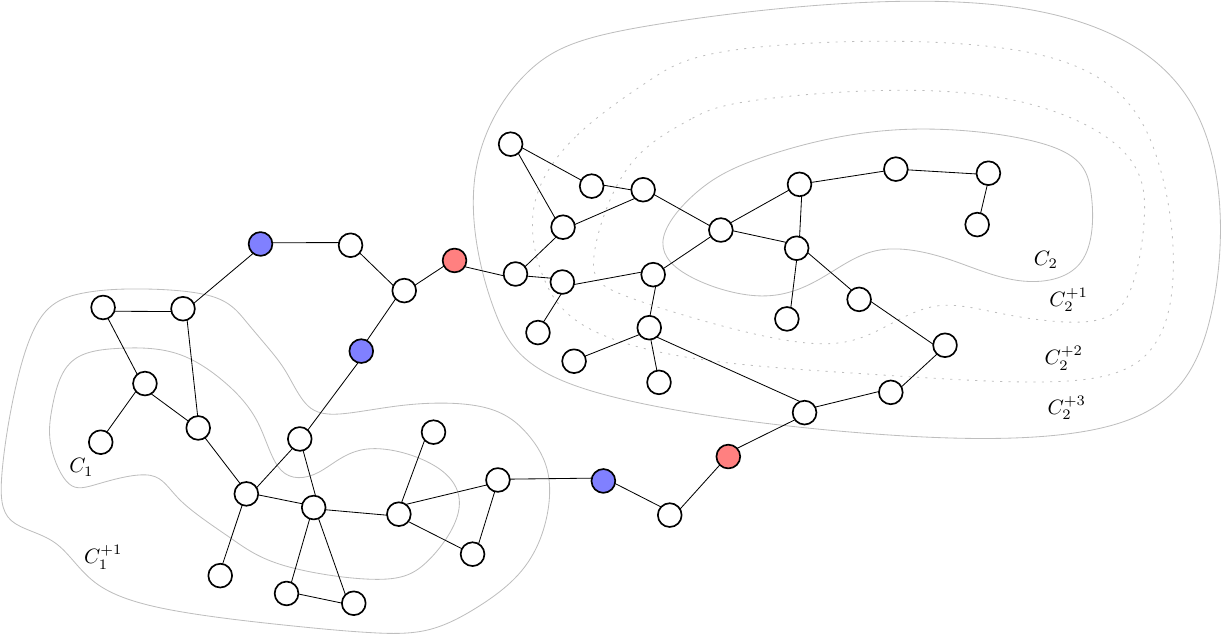}
    \caption{The figure shows two clusters $C_1, C_2$ that are $8$-separated. We can enlarge $C_1$ to $C_1^{+1}$ and $C_2$ to $C_2^{+3}$ so that the new clusters $C_1^{+1}$ and $C_2^{+3}$ have the property that the number of nodes neighboring with them (blue for $C_1^{+1}$ and red for $C_2^{+3}$) is only a small fraction of their size. The clusters remain $4$-separated, and their diameter is only increased additively by at most $+2$ in case of $C_1$ and at most $+6$ in case of $C_2$. }
    \label{fig:ldc}
\end{figure*}

\begin{lemma}
\label{lem:unweighted_contraction}
Assume that for any $t'$ one can construct a $t'$-separated low diameter clustering of an $n$-vertex graph with diameter $\Dldc(n, t')$ in $\Tldc(n, t')$ rounds of the deterministic (randomized) \congest model. 

Then, for any $t > 0$ there is a deterministic (randomized) distributed algorithm that outputs a complete clustering $\fC$ together with a set of edges $E^{inter} \subseteq E(G)$ such that

\begin{enumerate}
    \item $\diam(\fC) \le \Dldc(n,10t) + 10t$, 
    \item $|E^{inter}| \le n /t $,
    \item for any two neighboring clusters $C_1,C_2 \in \fC$, there exists at least one edge in $E^{inter}$ with one endpoint in $C_1$ and one endpoint in $C_2$. 
\end{enumerate}
The round complexity of the algorithm in the \congest model is
\[
O(\log n) \cdot  \left( T(n,10t) + D(n,10t) + t \right).
\]
\end{lemma}

\begin{proof}

Let $T = \lceil 1+ \log_{10/7} n \rceil$. 
The algorithm computes a sequence of clusterings of $G$ $\fC_0 \subseteq \fC_1 \subseteq \dots \subseteq \fC_T$ and a sequence of edges $E^{inter}_0 \subseteq E^{inter}_1 \subseteq \ldots \subseteq E^{inter}_T \subseteq E$. We denote by $V^{unclustered}_i$ the set of unclustered nodes in $\fC_i$. 
The clustering $\fC_i$ and the set of edges $E^{inter}_i$ will satisfy

\begin{enumerate}
    \item $\diam(\fC_i)\leq D(n,10t) + 10t$,
    \item $|V^{unclustered}_i| \leq (7/10)^i n$,
    \item $|E^{inter}_i| \leq |\bigcup_{C \in \fC_i} C|/t$,
    \item for any two neighboring clusters $C_1,C_2 \in \fC_i$, there exists at least one edge in $E^{inter}_i$ with one endpoint in $C_1$ and one endpoint in $C_2$,
    \item for  any node $v \in V^{unclustered}_i$ and neighboring cluster $C \in \fC_i$, there exists a node $u \in C$ with $\{v,u\} \in E^{inter}_i$.
\end{enumerate}

Note that the invariants for $i = 0$ are satisfied by letting $\fC_0$ be the empty clustering and $E^{inter}_0$ be the empty set.

Moreover, $|V^{unclustered}_T| \leq (7/10)^T n < 1$ and therefore $\fC_T$ is a complete clustering of $G$.
Hence, we can output $\fC = \fC_T$ and $E^{inter} = E^{inter}_T$.

Thus, it remains to show how to compute $\fC_{i+1}$ and $E^{inter}_{i+1}$ given $\fC_i$ and $E^{inter}_i$ while preserving the invariants.

Let $G_i = G[V^{unclustered}_i]$. First, we compute a $10t$-separated strong-diameter clustering $\fC'_i$ of $G_i$ with diameter $D(n,10t)$ that clusters at least $|V(G_i)|/2$ nodes in $T(n,10t)$ rounds 
of the deterministic (randomized) \congest model.

To explain the next part of the phase, let us introduce a new notation (see also \cref{fig:ldc}). For a cluster $C \in \fC'_i$, we denote by $C^{+j}$ its superset such that $u \in C^{+j}$ if and only if $d_{G_i}(u, C) \le j$.

For $j \in \{0,1,\ldots,4t-1\}$, we say that distance $j$ is a good cutting distance for a cluster $C \in \fC'_i$ if $C^{+j}$ has at most $\frac{|C|}{t}$ neighboring nodes in $G_i$.

We say that a cluster $C \in \fC'_i$ is good if there exists a good cutting distance for $C$, and otherwise, we refer to $C$ as bad.
For a good cluster $C \in \fC'_i$, let $j_C$ denote the smallest good cutting distance for $C$. 
Then, we obtain the clustering $\fC_{i+1}$ from $\fC_i$ by adding for each good cluster $C \in \fC'_i$ the cluster $C^{+j_C}$ to the clustering $\fC_{i+1}$.

Moreover, we obtain $E^{inter}_{i+1}$ from $E^{inter}_i$ by adding for each node $v$ in $G_i$ and newly added neighboring cluster $C^{+j_C}$ an arbitrary edge connecting $v$ with $C^{+j_C}$.

We next show that all the invariants are preserved, and afterward discuss the distributed implementation of the procedure.

We start with the first invariant, namely that $\diam(\fC_{i+1})\leq D(n,10t) + 10t$. Each cluster $C \in \fC'_i$ has a diameter of at most $D(n,10t)$. 
Therefore, $C^{+j}$ has a diameter of at most $D(n,10t) + 2j$ and as for each good cluster $C \in \fC'_i$, $j_C < 4k$, it directly follows that $\diam(\fC_{i+1})\leq D(n,10t) + 10t$, as needed.

Next, we have to show that $|V^{unclustered}_{i+1}| \leq (7/10)^{i+1} n$. As $|V^{unclustered}_{i}| \leq (7/10)^{i} n$, 
it therefore suffices to show that $|V^{unclustered}_{i+1}| \leq (7/10)|V^{unclustered}_{i}|$.
Note that $\fC'_i$ clusters at least $\frac{|V^{unclustered}_{i}|}{2}$ nodes. Hence, it suffices to show that at most one-fifth  of the nodes in $G_i$
are contained in bad clusters. 
As $\fC'_i$ is $10t$-separated, any two clusters $C_1 \neq C_2 \in \fC'_i$ satisfy $d_{G_i}(C_1^{+4t}, C_2^{+4t}) \ge 2t$. 
In particular, $C_1^{+4t}$ and $C_2^{+4t}$ are disjoint and therefore $|\bigcup_{C \in \fC'_i} C^{+4t}| \leq |V^{unclustered}_{i}|$.

Moreover, for each bad cluster $C$ we have $|C^{+4t}| \geq |C| + 4t \cdot \frac{|C|}{t} = 5|C|$. Therefore,

\[|\bigcup_{C \in \fC'_i \colon \text{$C$ is bad}} C| \leq \frac{1}{5} |\bigcup_{C \in \fC'_i} C^{+4t}| \leq \frac{1}{5}|V^{unclustered}_{i}|,\]

as needed. 

For the third invariant, we have to show that $|E^{inter}_{i+1}| \leq |\bigcup_{C \in \fC_{i+1}} C|/t$. As $|E^{inter}_i| \leq |\bigcup_{C \in \fC_{i}} C|/t$,
it suffices to show that $|E^{inter}_{i+1} \setminus E^{inter}_i| \leq |\bigcup_{C \in \fC_{i+1} \setminus \fC_i} C|/t$. 
Note that each newly added cluster $C$ is neighboring with at most $|C|/t$ nodes. As $E^{inter}_{i+1} \setminus E^{inter}_i$ contains at most one edge
with one endpoint in $C$ for each neighboring vertex $v$, this implies that $E^{inter}_{i+1} \setminus E^{inter}_i$ contains at most $|C|/t$ edges with one 
endpoint in $C$, which shows that the third invariant is preserved.

Next, we verify that the fourth invariant is preserved. Consider any two neighboring clusters $C_1,C_2 \in \fC_{i+1}$. We have to show that there exists
at least one edge in $E^{inter}_{i+1}$ with one endpoint in $C_1$ and one endpoint in $C_2$. If both $C_1$ and $C_2$ were already contained in $\fC_i$, then this
property directly follows from induction. Note that $C_1$ and $C_2$ cannot both be newly added clusters, as this would imply that $d_{G_i}(C_1,C_2) > 2t$ and therefore
they could not be neighboring in $G$. Thus, it remains to consider the case that exactly one of the two clusters was newly added to $\fC_{i+1}$, let's say $C_1$.
As $C_1$ and $C_2$ are neighbors, there exists by definition a vertex $v$ in $C_1$ that is neighboring with $C_2$. As $C_1$ is newly added, $v$ was unclustered before,
i.e., $v \in V^{unclustered}_i$. In particular, the fifth invariant implies that $E^{inter}_i$ contains an edge incident to $v$ with the other endpoint being contained in
$C_2$, as desired.

We now check the last invariant. Consider an arbitrary $v \in V^{unclustered}_{i+1}$ and cluster $C \in \fC_{i+1}$ such that $v$ is neighboring with $C$. 
We have to show that there exists a node $u \in C$ with $\{v,u\} \in E^{inter}_{i+1}$. If $C \in \fC_i$, then this follows from induction.
Otherwise, it directly follows from the description of how we obtain $E^{inter}_{i+1}$ from $E^{inter}_i$.

It remains to discuss the distributed implementation. 

Computing the clustering $\fC'_i$ takes $\Tldc(n, 10t)$ \congest rounds.
Next, we have to decide for each cluster $C \in \fC'_i$ whether it is good or bad, and if $C$ is a good cluster, we additionally have to compute the smallest good cutting distance of $C$.
For each $C \in \fC'_i$, let $F_C$ be a BFS forest in $G_i$ of depth $4t$ with $C$ being the set of roots. As $\fC'_i$ is $10t$-separated, it follows that $F_C$ is disjoint from $F_{C'}$ 
for any other cluster $C' \in \fC'_i$. Hence, we can compute $F_C$ for all the clusters $C \in \fC'_i$ in $O(t)$ \congest rounds.
Moreover, each node $u$ in $F_C$ gets to know its distance $d_{G_i}(u,C)$ to $C$ in $G_i$.

A standard pipelining idea therefore implies that each root of $F_C$ can learn for every $j \in \{0,1,2,\ldots,4t\}$ how many nodes of depth $j$ its tree contains in $O(t)$ \congest rounds.
As $\diam(C) \leq D(n,10t)$, with another standard pipelining idea, the root of the cluster $C$ can learn in $O(t + D(n,10t))$ additional \congest rounds for each $j \in \{0,1,2,\ldots,4t\}$
how many vertices $u \in V(G_i)$ with $d_{G_i}(u,C) = j$ exists. This also allows the root to decide whether $C$ is good or bad, and if it is good, to compute the smallest good cutting distance of $C$.

Finally, the root can broadcast this information to each node having a distance of at most $4t$ to $C$ in $O(t + D(n,10t))$ \congest rounds.

Moreover, after having computed the clustering $\fC_{i+1}$, it is easy to see that we can compute $E^{inter}_{i+1}$ in $O(1)$ additional \congest rounds.

Hence, the overall round complexity of the algorithm is $O(\log n) \cdot (T(n,10t) + D(n,10t) + t)$, as desired.

\end{proof}

\paragraph{The Spanner Construction}
We are now ready to prove \cref{thm:unweighted_spanner_warmup}.

\begin{proof}
We apply \cref{lem:unweighted_contraction} to get a complete clustering $\fC$ and a set of edges $E^{inter}$. For each cluster $C \in \fC$, let $T_C$ be a spanning tree of $C$ with diameter $O(\Dldc(n,10t) + t)$. Such a spanning tree can be computed by a simple BFS from an arbitrary node in the cluster, as the diameter of each cluster is $O(\Dldc(n,10t) + t)$. Now, let $F$ be the union of trees $T_C$ for $C \in \fC$. 

The spanner $H$ now includes all the edges of $F$ and all the edges of $E^{inter}$.

As $F$ is a forest, it contains at most $n-1$ edges and as $|E^{inter}| \leq \frac{n}{t}$, we get $|E(H)| \leq n + n/t$. 
\begin{claim}
\label{cl:unweighted_diam}
$H$ is an $O(\Dldc(n,10t) + t)$-spanner. 
\end{claim}
\begin{proof}
Consider any edge $\{u,v\} \in E(G)$. We have to show that $d_H(u,v) = O(\Dldc(n,10t) + t)$.  

First, consider the case that $u$ and $v$ belong to the same cluster. 
Hence, there exists a path of length $O(\Dldc(n,10t) + t)$ between $u$ and $v$ in $F$ and therefore $d_H(u,v) = O(\Dldc(n,10t) + t)$. 

It remains to consider the case that $u \in C_1$ while $v \in C_2 \not= C_1$. 
As $C_1$ and $C_2$ are neighbors, there exits at least one edge $\{x,y\}\in E^{inter}$ with $x \in C_1$ and $y \in C_2$.
Now, $u$ and $x$ are connected by a path of length $O(\Dldc(n,10t) + t)$ in $F$ and $v$ and $y$ are connected by a path of length $O(\Dldc(n,10t) + t)$ in $F$.
Thud, $u$ and $v$ are connected by a path of length at most $2 \cdot O(\Dldc(n,10t) + t) + 1 = O(\Dldc(n,10t) + t)$ and therefore $d_H(u,v) = O(\Dldc(n,10t) + t)$, as desired.

\end{proof}

It remains to analyze the complexity.
The algorithm invokes \cref{lem:unweighted_contraction} with round complexity $O(\log n) \cdot \left( \Tldc(n, 10t) + \Dldc(n,10t) + t \right)$.

After that, we need $O(\diam(\fC)) = O(\Dldc(n, 10t) + t)$ additional rounds to compute the forest $F$ and therefore the construction of the whole spanner takes 
$O(\log n) \cdot \left( \Tldc(n, 10t) + \Dldc(n,10t) + t \right)$ \congest rounds.
\end{proof}

\subsection{Unweighted Sparse Spanners via Weak-Diameter Clusterings}
\label{sec:weak_implies_sparse}

In this section, we prove \cref{thm:final-result-weak-clustering-to-utrasparse-unweighted}. It shows that an efficient algorithm for computing a $3$-separated weak-diameter clustering, defined below, implies an efficient algorithm for computing a sparse unweighted spanner. 

\begin{definition}[$t$-separated Weak Diameter Clustering]
\label{def:weak_diameter_clustering}
Assume we are given an unweighted graph $G$ and parameters $t, \xiavg, D$. A $t$-separated weak-diameter clustering with weak-diameter $D$ is a clustering $\fC$, with each cluster $C$ coming with a tree $T_C, V(T_C) \supseteq C$ such that

\begin{enumerate}
    \item For each $C \in \fC$ we have $\diam(T_C) \le D$. 
    \item For each $C_1 \not = C_2$ we have $d(C_1, C_2) \ge t$.
    \item We have $|\bigcup_{C \in \fC} C| \ge n/2$. 
\end{enumerate}

Moreover, if $\xi(v)$ is the number of trees $T_C, C \in \fC$ that a node $v \in G$ is contained in, then  $\sum_{v \in V(G)} \xi(v) \le n \cdot \xiavg$. 
\end{definition}

Note that each strong-diameter clustering is a weak-diameter clustering with the same parameters and $\xiavg \leq 1$. Therefore, \cref{thm:final-result-weak-clustering-to-utrasparse-unweighted} can be seen as a stronger result compared to \cref{thm:unweighted_spanner_warmup} in the sense that it only assumes an efficient algorithm for computing a weak-diameter clustering instead of a strong one. On the other hand, \cref{thm:final-result-weak-clustering-to-utrasparse-unweighted} cannot be used to directly obtain ultra-sparse spanners, which is however not a problem due to \cref{thm:p1-sparse-ultrasparse}.

\thmDeterministic*

\begin{proof}
The proof is very similar to the proof of \cref{thm:unweighted_spanner_warmup}. 
We start with an empty clustering. In each step, we add additional clusters to the clustering until we have a complete clustering.
In particular, consider a fixed step and let $V^{unclustered}$ denote the set of unclustered vertices with respect to the current clustering.
We now use algorithm $\fA$ to compute a $3$-separated clustering $\fC$ in $G[V^{unclustered}]$ with weak-diameter $D(n)$ and average overlap at most $\xiavg$ clustering at least half of the nodes.
For each cluster $C \in \fC$, we add the cluster $C$ to the current clustering and we add all edges in the tree $T_C$ to the final spanner. This way, we add at most
$\xiavg(n) \cdot |V^{unclustered}|$ many edges to the spanner. Moreover, for each unclustered
node $v$ neighboring with a cluster $C \in \fC$ (due to the $3$-separation, each node neighbors with at most one cluster), we add an arbitrary edge from $v$ to $C$ to the final spanner, for a total of at most $|V^{unclustered}|$ edges.  As the number of unclustered nodes decreases by a factor of $2$ in each step, all the nodes are clustered after $O(\log n)$ steps
and the total number of edges of the final spanner is $O(\xiavg(n)n)$.
With the exact same argumentation as in the proof of \cref{cl:unweighted_diam}, it follows that $H$ is a $\beta$-spanner with $\beta = O(D(n))$.
The runtime bounds readily follow.
\end{proof}

\subsection{Weighted Work-Efficient Ultra-Sparse Spanners}
\label{sec:putting_everything_together}

In this section, we show how to efficiently compute weighted ultra-sparse spanners. 

\thmWeightedUltra*
\begin{proof}
The starting point is the following theorem.

\begin{theorem}[Theorem 1.12 in \cite{rozhovn2020polylogarithmic}]
\label{thm:weak_diam}
Given an undirected graph, we can build a weak-diameter clustering $\fC$ such that
\begin{enumerate}
    \item Each $T_C$ for a $C \in \fC$ satisfies $\diam(T_C) = O(\log^3 n)$. 
    \item Every two different clusters $C_1,C_2 \in \fC$ are at least $10$-separated. 
    \item Every node of $V(G)$ is in at most $O(\log n)$ trees $T_C$. 
    \item $|\bigcup \fC| \ge n/2$. 
\end{enumerate}
The algorithm needs $O(\log^6 n)$ deterministic \congest rounds.  \footnote{The version of the theorem in~\cite{rozhovn2020polylogarithmic} has an additional log factor as it builds a network decomposition instead of a single clustering. }
\end{theorem}
We remark that the algorithm of \cite{rozhovn2020polylogarithmic} also works when each node is actually a cluster of diameter at most $r$, in which case the round complexity is $O(r \log^6 n)$. Therefore, \cref{thm:final-result-weak-clustering-to-utrasparse-unweighted} implies that we can compute an unweighted spanner on such a graph with $O(n \log n)$ edges and stretch $O(\log^3 n)$ in $O(r \log^7 n)$ \congest rounds. Hence, the folklore reduction from weighted to unweighted implies that we can compute a weighted spanner with $O(n \log n \log (U + 1))$ edges and stretch $O(\log^3 n)$ in $O(r \log^8 n)$ \congest rounds (as we are on a cluster graph, we need to compute the $O(\log (U+1))$ unweighted spanners one after the other). This algorithm together with the reduction from ultra-sparse to sparse (\cref{thm:p1-sparse-ultrasparse}) then implies the \congest  part of \cref{thm:result-weighted-ultra-randomized-congest}.
The PRAM bound follows from the fact that (1) the algorithm of \cite{rozhovn2020polylogarithmic} can be implemented in near-linear work and polylogarithmic time, (2) the sparse spanner can therefore be computed in near-linear work and polylogarithmic depth, and (3) the reduction from sparse to ultra-sparse can be implemented in near-linear work and polylogarithmic time. 

\end{proof}


\section{Connectivity Certificates}
\label{sec:conn-certificate}
A simple and common way to find a $k$-connectivity certificate of a graph is to repeatedly extract a skeleton from $G$ for $k$ times. More concretely, in the $i$-step, we remove a skeleton $H_i$ from $G \setminus \cup_{j=1}^{i-1} H_j$. At the end, $H = \cup_{i=1}^{k} H_i$ should be a $k$-connectivity certificate of $G$. Indeed, the output provides a stronger guarantee: It contains either all or at least $k$ edges from each cut in $G$. For the proof of this, consider an arbitrary cut. If each $H_i$ has an edge in the cut, then $H$ has at least $k$ such edges as $H_i$s are disjoint. So suppose there is an $i$ such that $H_i$ has no edge in the cut. Since $H_i$ is a skeleton, there should be no such edge in $G\setminus \cup_{j=1}^{i-1} H_j$, implying that all edges in this particular cut should be in $\cup_{j=1}^{i-1} H_j$.

This algorithm gives us a certificate with size at most $(n-1)k$ if we extract a forest at each step. Finding a forest in the distributed setting needs $\Omega(D + \sqrt{n})$ many rounds, where $D$ denotes the diameter of network $G$. Instead of extracting a forest per step, we can extract an ultra-sparse spanner at each step, as any spanner is also a skeleton by definition. Using \Cref{thm:ultrasparse}, we can find a spanner with size $n(1 + \eps)$ in $\frac{\polylog(n)}{\eps}$ rounds.

\begin{theorem}
\label{thm:certificate-smallk}
There is a deterministic distributed algorithm that computes a $k$-connectivity certificate of $G$ with at most $nk(1 + \eps)$ edges in $\frac{k\polylog(n)}{\eps}$ rounds. Moreover, the output has either all or at least $k$ edges of every cut in $G$.
\end{theorem}

As the $k$-connectivity certificate of each $k$-edge-connected graph has at least $k n/2$ edges, the theorem above gives us a $2(1 + \eps)$ approximation. The problem though is that the number of rounds is linear in $k$. To overcome this issue, we use random sampling of Karger~\cite{karger1999random} and give a randomized algorithm that runs in $\polylog(n)$ rounds for any $k$.

According to~\cite[Theorem 2.1]{karger1999random}, if we sample each edge of a $k$-edge-connected graph with probability $p = O(\frac{\log n}{k \eps^2})$, then with high probability, the number of sampled edges in every cut is around its expectation up to a factor $\eps$. More concretely, for a cut with $c$ edges, at least $pc(1-\eps)$ and at most $pc(1 + \eps)$ edges are sampled.

Suppose $G$ is $k$-edge-connected. We split its edges randomly into $Q = \Theta(\frac{k \eps^2}{\log n})$ partitions. So with high probability, every cut in $G$ with size $c$ has size $\frac{c}{Q}(1 \pm \eps)$ edges in every partition (Since there are $Q = O(n)$ partitions, we can take union bound on $Q$ applications of the sampling theorem). 

We compute a $\left(k' = \frac{k(1 + \eps)}{Q(1 - \eps)}\right)$-connectivity certificate for each partition simultaneously using \cref{thm:certificate-smallk}. This takes $\frac{k'\polylog(n)}{\eps} = \frac{\polylog(n)}{\eps^3}$ many rounds. We show that the union of these certificates is a $k$-connectivity certificate for $G$. For this, consider an arbitrary cut $C$ with $|C|$ edges. Recall that $G$ is $k$-edge-connected, so $|C| \geq k$. To show that at least $k$ edges of $C$ is in the union of these $Q$ certificates, we consider two cases based on the size of $C$:
\begin{itemize}[-]
    \item If $|C|$ is at most $\frac{k}{1 - \eps}$, then each partition contains at most $\frac{|C|(1 + \eps)}{Q} \leq \frac{k(1 + \eps)}{Q(1 - \eps)} \leq k'$ edges of $C$. From \cref{thm:certificate-smallk}, all of them should appear in the certificate of that partition and so, all the edges of $C$ is in our output. 
    \item If $|C|$ is more than $\frac{k}{1 - \eps}$, then each partition has at least $\frac{|C|(1 - \eps)}{Q} \geq \frac{k}{Q}$ edges of $C$. Since $k' \geq \frac{k}{Q}$, the $k'$-connectivity certificate of each partition contains at least $\frac{k}{Q}$ from $C$ implying that there is at least $k = \frac{k}{Q} \cdot Q$ edges of $C$ in our output.
\end{itemize}
This concludes that if $G$ is $k$-edge-connected, then the union of these $Q$ certificates is also $k$-edge-connected. If $G$ is not $k$-edge-connected, then obviously any subgraph of it is not $k$-edge-connected. So our output is always a $k$-connectivity certificate of $G$. The number of edges in our certificate is $Q\cdot(nk'(1 + \eps)) = nk\frac{(1 + \eps)^2}{1-\eps}$. Assuming $\eps < \frac{1}{2}$, we have $\frac{1}{1 - \eps} \leq (1 + 2\eps)$ and $(1 + \eps)^2 \leq (1 + 3\eps)$. So the total number of edges in the output is $nk(1 + 8\eps)$. Replacing $\eps$ with $\frac{\eps}{8}$ proves \cref{thm:certificate-largek}.